\documentclass[preprint,12pt,3p]{elsarticle}
\usepackage{amsmath,amsthm,amsfonts,amssymb}
\usepackage{fullpage}
\usepackage{graphicx}

\usepackage{spverbatim}

\usepackage{multicol}
\usepackage{blkarray}
\usepackage{caption}

\usepackage[tight,footnotesize]{subfigure}
\usepackage{float}
\usepackage{algpseudocode}
\usepackage{tikz}
\usetikzlibrary{trees,arrows}
\usepackage[toc]{appendix}

\usepackage{hyperref}
\numberwithin{equation}{section}
\usepackage{xcolor}
\usepackage{colortbl}
\usepackage[normalem]{ulem}
\usepackage{sectsty}

\definecolor{astral}{RGB}{46,116,181}
\subsectionfont{\color{astral}}
\sectionfont{\color{astral}}
\usepackage{colortbl}
\usepackage[outercaption]{sidecap}

\DeclareMathAlphabet{\mathpzc}{OT1}{pzc}{m}{it}
\DeclareFontFamily{OT1}{pzc}{}
\DeclareFontShape{OT1}{pzc}{m}{it}{<-> s * [0.900] pzcmi7t}{}
\DeclareMathAlphabet{\mathpzc}{OT1}{pzc}{m}{it}

\usepackage{amsbsy}
\usepackage{amsmath}
\usepackage{accents}
\newlength{\dhatheight}

\usepackage{hyperref}
\hypersetup{
    colorlinks=true,
    linkcolor=blue,
    filecolor=magenta,      
    urlcolor=cyan,
}

\tikzset{
    onen/.style={
        draw,
        text=black,
        fill=red,
        minimum height=1.5cm,
        minimum width=3cm},
     twoo/.style={
        draw,
        text=black,
        fill=orange,
        minimum height=1.5cm,
        minimum width=3cm},
    thrree/.style={
        text=black,
        draw,
        minimum height=1.5cm,
        minimum width=3cm,
        left color=green, right color=green},
    fouur/.style={
        draw,
        text=black,
        fill=orange,
        minimum height=1.5cm,
        minimum width=4cm}}    
    
\newenvironment{frcseries}{\fontfamily{frc}\selectfont}{}

\DeclareMathAlphabet\mathbfcal{OMS}{cmsy}{b}{n}
\definecolor{darkslategray}{rgb}{0.18, 0.31, 0.31}
\definecolor{warmblack}{rgb}{0.0, 0.26, 0.26}

\usepackage{multirow}
\usepackage{mathtools}
\usepackage{algorithm}
\usepackage{algorithmicx}
\makeatletter
\def\BState{\State\hskip-\ALG@thistlm}
\makeatother
\newtheorem{theorem}{Theorem}[section]

\theoremstyle{definition}
\usepackage{hyperref}
\newtheorem{definition}{Definition}[section]

\begin{document}

\begin{frontmatter}

\title{ \textcolor{warmblack}{\bf Modeling Control, Lockdown \& Exit Strategies for COVID-19 Pandemic in India
}}

\cortext[cor1]{Corresponding authors}

\author[label1]{Madhab Barman}
\address[label1]{Department of Mathematics, Indian Institute of Information Technology Design and Manufacturing Kancheepuram, Chennai-600127, India}

\author[label1]{Snigdhashree Nayak}

\author[label2]{Manoj K. Yadav}
\address[label2]{Department of Mathematics, \'{E}cole Centrale School of Engineering, Mahindra University, Hyderabad - 500043, India}

\author[label3]{Soumyendu Raha\corref{cor1}}\address[label3]{Department of Computational and Data Sciences, Indian Institute of Science, Bangalore-560012, India}\ead{raha@iisc.ac.in}

\author[label1]{Nachiketa Mishra\corref{cor1}}
\ead{nmishra@iiitdm.ac.in}

\begin{abstract}
COVID-19--a viral infectious disease--has quickly emerged as a global pandemic infecting millions of people with a significant number of deaths across the globe. The symptoms of this disease vary widely. Depending on the symptoms an infected person is broadly classified into two categories namely, asymptomatic and symptomatic. Asymptomatic individuals display mild or no symptoms but continue to transmit the infection to otherwise healthy individuals. This particular aspect of asymptomatic infection poses a major obstacle in managing and controlling the transmission of the infectious disease. In this paper, we attempt to mathematically model the spread of COVID-19 in India under various intervention strategies. We consider SEIR type epidemiological models, incorporated with India specific social contact matrix representing contact structures among different age groups of the population. Impact of various factors such as presence of asymptotic individuals, lockdown strategies, social distancing practices, quarantine, and hospitalization on the disease transmission is extensively studied. Numerical simulation of our model is matched with the real COVID-19 data of India till May 15, 2020 for the purpose of estimating the model parameters. Our model with zone-wise lockdown is seen to give a decent prediction for July 20, 2020.    
\end{abstract}


\end{frontmatter}


\section{Introduction}\label{sec:01}
COVID-19 is an infectious disease caused by the novel human coronavirus named SARS-CoV-2. Towards the end of December 2019, the first few cases of the COVID-19 outbreak was reported as a mysterious pneumonia from Wuhan in the Hubei Province of China. In a short span of time, this infectious disease has rapidly spread across the world and has catapulted into a global pandemic. As of 15th July 2020, the virus has infected over 13.5 million people with more than 580,000 confirmed deaths globally. \\

\noindent The symptoms of SARS-CoV-2 virus infection vary widely with most people only experiencing mild to moderate respiratory illness symptoms. A small proportion of people develop severe respiratory complications often requiring ICU and ventilator support \cite{India:critical}, except few regions \cite{Seattle:critical, JAMA:2020} like Seattle and Lombardy (Italy), where the percentage is higher. Based on clinical data, elderly people have been found to be at greater risk of experiencing acute respiratory distress symptoms with high mortality rate as compared to people of younger age. Study also suggests that approximately $80\%$ of the infected individuals are asymptomatic carriers who experience mild or no symptoms but continue to transmit the virus to otherwise healthy people. This has caused the detection and containment of the SARS-CoV-2 virus transmission a very challenging problem for civic authorities.\\

\noindent The COVID-19 pandemic is inflicting significantly high mortality, straining public healthcare systems and causing severe socio-economic distress globally. In the absence of any potent vaccine or effective pharmaco-medical treatments available, all efforts towards the pandemic management and mitigation have largely focused on non-pharmaceutical interventions like social distancing, lockdowns, contact tracing, quarantine and isolation. Intensive testing, contact tracing, and isolation of cases has to a large extent enabled disease transmission control in several places, such as Israel, Singapore and South Korea. Goal of these intervention strategies is to slowdown the disease transmission, reduce mortality rate and ameliorate the burden and strain on healthcare systems.\\

\noindent Since the seminal work \cite{SIR:1927} of Kermack-McKendric in 1927 on SIR epidemic model \cite{ross:1911}, several researchers have used adapted or modified versions of the basic SIR epidemiological model for modelling evolution of epidemics. Various versions of the basic models are essentially systems of first order ODEs with an incidence function and dependency on constant parameters describing the nature of the infectious disease such as rate of transmission on contact, rate of recovery, mortality, incubation period etc. In practice these parameters vary during the course of the epidemic and also across geographical regions and local population. Another drawback of SIR type models involving intervention strategies is eventually almost the whole susceptible population become infected in a short span of time due to the exponential rate of transmission. Intervention strategies applied on SIR models only delay the eventual infection of almost entire population.\\ 

\noindent Generally, epidemics with latency periods \cite{ Julian:2020, AMS:2008, Erlang:2018} are modeled by compartmental epidemiological models of the type Susceptible-Exposed-Infected-Removed (SEIR). Mathematics of epidemic models like SIR and SEIR, and its variants can be seen in \cite{SIAMreview:2000}, along with different threshold numbers.  Let $N$ denote the size of a population, $S(t)$ the number of susceptible individuals at time $t$, $E(t)$ the number of exposed individuals at time $t$, $I(t)$ the number of infected individuals at time $t$, and $R(t)$ the number of recovered and death cases at time $t$. Then evolution of the epidemic is governed by the system of first order ODEs
\begin{eqnarray}
&&\dot{S}(t)=\mu N -\beta SI/N -\mu S, \nonumber \\
&&\dot{E}(t)=\beta SI/N -(\mu+\alpha) E, \nonumber \\
&&\dot{I}(t)= \alpha E-(\mu+\gamma) I, \nonumber \\
&&\dot{R}(t)=\gamma I -\mu R. \nonumber
\end{eqnarray}
Here $\mu$ denote constant rate of new recruitment to the susceptible population as well as natural death rate in each compartment. Parameter $\beta$ measures potential force of infection on contact, whereas parameters $\alpha$ and $\gamma$ denote rate of transition from compartments $E$ and $I$ to $I$ and $R$, respectively. One may note that $\dot{S}(t)+\dot{E}(t)+\dot{I}(t)+\dot{R}(t)=0$.
Therefore, $S(t)+E(t)+I(t)+R(t)=N$ is constant for all time $t\geq 0$ and the solution space is
$\{(S,E,I,R)\in \mathbb{R}^4_+:S+E+I+R=N\}$.
The model is extended here along with the contact network to perform a data based analysis for the COVID-19 spread and its control in India.\\

\noindent \emph{Main Contributions }: The focus of our work is to understand the different disease control target interventions and accordingly predict the disease spread possibilities by estimating the model parameters from the available data. In this attempt our main contributions are listed below.
\begin{enumerate}
    \item We have used two epidemic models -- SEAIRD and control based SEAIRD -- to model the evolution of COVID-19 in India. The second model fine tunes the first model towards the twin goals of capturing the real situation and making credible suggestions to policy makers involved in the pandemic management and mitigation.
    
     \item The disease free and endemic stability of the SEAIRD model has been established.
    
    \item Presence of asymptomatic infections and their role in largely unbridled and quick disease transmission is computationally established.
    
    \item The control based SEAIRD model has two very interesting features. One being a reverse flow from quarantine class (Q) to susceptible class (S) and another one modeling increased awareness for social distancing practices among the population. These features model the effectiveness of the general and individual level preventive measures quarantine and social distancing.
    
    \item The control based SEAIRD model parameters are traced back by matching the real COVID-19 data of India for the first 93 days. Some model based predictions are presented to understand the possibility of second wave, its size and time of arrival under different lockdown policies. 
    
    \item Three lockdown policies are considered. One of them (LD-I) attempt to model different lockdown phases (complete lockdown, zone-wise lockdown) in India. Additionally, two staggered policies are designed as suggestions to policy makers for further bringing down the infection levels and mitigate the adverse socio-economic impact.
    
\end{enumerate}

\noindent Before presenting the technical details, in  Section \ref{sec:02}, we enlist some relevant definitions and theorems useful in describing mathematical models of epidemics. A standard model suitable for modeling COVID-19 pandemic named as SEAIRD is proposed in Section \ref{sec:03}, incorporating the age and contact structures to track the COVID-19 evolution in India. In this model, the removed compartment (R) is further partitioned into two compartments, namely, recovered class (R) and death class (D). The infective class of population is subdivided into two classes, namely, asymptomatic (A) and symptomatic (I). Presence of asymptomatic infections and their role in largely unbridled and quick disease transmission is computationally established. There are some theoretical support for the new SEAIRD model discussed in Subsection \ref{sec:03:stab} for establishing the disease free and endemic equilibrium points and their stability. We study the disease transmission under three different lockdown scenarios namely, no lockdown, strict lockdown for prolonged time period, and staggered lockdown phases of varying degrees of implementation.\\

\noindent An improvised SEAIRD model where the scope of the pandemic model is broadened to incorporate quarantine (Q) and hospitalization (H) measures is considered in Section \ref{sec:04} to further control the disease transmission. The model parameters are traced back by comparing with the real COVID-19 data in India. In Section \ref{sec:05}, multiple lockdown policies and social distancing strategies are experimented computationally, which suggest the possibility of second wave of infections and possible optimal control strategy in next two years time frame. This analysis is followed by limitations of the model, future scope and challenges discussion in Subsection \ref{sec:05:challenge}. Finally, Section \ref{sec:06}, presents the conclusion of this work.


\section{Epidemic Model : Basic Definitions \& Concepts}\label{sec:02}
The general pandemic model considers a heterogeneous population that can be grouped into $n$ homogeneous compartments. Further, the compartments may be sorted such that the first $m$ compartments consists of all the infected cases. Let $x=(x_1, \cdots ,x_n)^T$ represent a general state of the model where, $x_i \geq 0$ denote the number of individuals in each compartment. Thus, all states of the model are restricted to the closed positive cone $x \in X = \mathbb{R}^+$. Further, let
$$X_s = \{ x\geq 0\; | \;x_i = 0, i = 1, 2 \cdots, m \}$$
denote the set of disease free states. The system of differential equations modeling disease transmission in $X$ is of the form
\begin{equation}\label{dyn}
   \dot{x} = f(x),~~f=(f_1,f_2,\cdots,f_n) 
\end{equation}
where components of $f$ are of the form $f_i(x) = F_i(x) - W_i(x), i = 1, 2, \cdots, n$. Here $F_i(x)$ represents the rate of appearance of new infections in compartment $i$ by all other means and $$W_i(x) = W_i^{-}(x) - W_i^{+}(x),$$ where $W_i^{-}(x)$ is the rate of transfer of individuals out of the $i$th compartment and $W_i^{+}(x)$ is the rate of transfer of individuals into compartment $i$ by all other means. It is assumed that each function is at least twice continuously differentiable in each variable. Furthermore, it is required that $F \geq 0$ and $W$ is an $M$-matrix, so that, $W^{-1} \geq 0$. The next-generation matrix \cite{NGM:2008} is computed corresponding to the $m$ infected compartments and defined to be $K = FW^{-1} \geq 0$. The next-generation matrix is used in computing the reproduction number.

\begin{definition}\textnormal{(\cite{Rzero:1990})}\label{defR0}\\
In epidemiology, we take basic reproduction number/ratio,
$R_0$, as the average number of individuals infected by a
single infected individual during his or her entire
infectious period, in a population which is entirely
susceptible.
\end{definition}
\noindent The reproduction number can be computed via different methods. We have taken the next-generation matrix approach to compute the reproduction number as spectral radius of the next-generation matrix $K$, that is, $$R_0 = \rho(FW^{-1}).$$

\begin{definition}\textnormal{(\cite{Chaos:2000})}\label{def:stab}\\
A point $x^*$ is said to be an \emph{equilibrium point} of $f$, if $f(x^*) = 0$. An equilibrium point $x^*$ is \emph{stable}, if every initial point $x_0$ which is close to $x^*$ has the property: the solution $F(t, x_0)$ remain close to $x^*$ for all $t \geq 0$. An equilibrium point $x^*$ is called \emph{asymptotically stable} if it is both stable and attracting. The point $x^*$ is unstable if it is not stable. An equilibrium point $x^*$ is called \emph{globally   asymptotically stable} if it is asymptotically stable and for all initial values converge to the equilibrium point.
\end{definition}
\noindent Criteria for stability of a system of linear/linearized differential equation is given by
\begin{theorem}\textnormal{(Theorem 7.2, \cite{Chaos:2000})}\label{th:stab}\\
Let $\displaystyle{J(x^*)}$ be the $n \times n$ Jacobian matrix, corresponding to the system of equations $\dot{x} = f(x)$. If the real parts of all eigenvalues of $\displaystyle{J(x^*)}$ are negative, then the equilibrium point $x^*$ is globally asymptotically stable. If $\displaystyle{J(x^*)}$ has $n$ distinct eigenvalues and if the real parts of all eigenvalues of $\displaystyle{J(x^*)}$ are non-positive, then $x^*$ is stable.
\end{theorem}


\subsection{Social contact matrix}\label{sec:02:contact}
We intend to mathematically model the spread of COVID-19 in India under the influence of various intervention strategies and control measures. In order to account for the heterogeneous and local in nature contact patterns among various age groups, we consider a India specific social contact matrix of $16$ age groups based on the Demographic and Health Survey (DHS) data. Due to significant presence of asymptomatic individuals, there is a very high potential of rapid disease transmission. In view of mild or no symptoms, the asymptomatic population would follow normal social contact pattern. We incorporate the social contact matrix in our models for more realistic modelling of the disease transmission and the impact of other control strategies. \\

\noindent Research using social-contact networks has shown its efficiency in measuring the transmission scale, and analysing the relative merits and effectiveness of several proposed mitigation and intervention strategies \cite{nature:2004}. Early epidemiological models were based on
population-wise random-mixing, but in practice, each individual has a finite set of contacts to whom they can pass infection; whereas, the essential service providers are prone to higher degrees of contact out of which some of them may turn out to be super-spreaders \cite{superspread:2005, superspread:2020}. Knowledge of network structures allow models to estimate the epidemic dynamics at the population scale rather than individual level. Several methods that allow mixing of network or network approximation are reviewed in \cite{keeling:2005}. Detailed mathematical models characterizing early epidemic growth patterns incorporating inhomogeneous mixing of population networks are reviewed in \cite{chowell:2016}. The heterogeneous social contact networks are more likely to result in epidemic spreading than their homogeneous counterparts, thus having a major role in determining whether an infection would turn out to be an epidemic or persist at endemic levels \cite{heterogeneous:2020, PLoS:2011, sci_report:2015}. Therefore, epidemic models with interventions towards the goal of successfully preventing an outbreak need to account for social structure and mixing patterns. Contact patterns vary across age and locations (e.g. home, work, and school), therefore, integrating them with the transmission dynamics models of pathogens significantly improves the models? realism.\\

\noindent Let $C(t)$ denotes the contact matrix of $M$ individual age-groups at a certain time $t$. The $i\,j\mbox{-th}$ entry of the matrix $C(t)$ represent the number of contacts of an individual in age-group $i$ with another individual of age-group $j$. By the reciprocity relation, the number of contacts must satisfy 
    $$C_{i\,j} N_i = C_{j\,i} N_j, ~\mbox{where}~ 1\leq i,\,j \leq M.$$
The contact matrix $C$ consists of four component contact matrices, namely, workplace ($C^W$), home ($C^H$), school ($C^S$) and others ($C^O$). Thus, $C = C^W + C^H + C^S + C^O$. The contact structure data for India and its subdivisions are obtained from \cite{contact, IMSC:2020, cij:2012}.


\subsection{COVID-19 data source}\label{sec:02:data}
The growth in SARS-CoV-2 infection is recorded for the available Indian data set. The data may be influenced by different unavoidable constraints like variation in testing strategy and facilities of different states, nonuniform policies of state governments and public awareness levels. But our model parameters are basically data driven, and intended to predict and compute the total number of infections and death cases. The data available at \cite{covid19} are most reliable in terms of recording daily COVID-19 cases in India. This data has been used in forecasting future transmission scenario; determining the rate and extent of infection spread; and determining the longevity and level of lockdown measures along with the social distancing norms. Further, the available data with the mathematical model has been considered to understand the possibility of second wave of COVID-19 infection and its size.


\section{SEAIRD Epidemic Model} \label{sec:03}

Mathematical models have become important tools in analyzing the spread and control of infectious diseases. Furthermore, mathematical models have been used in comparing, planning, implementing, evaluating, and optimizing various detection, prevention, therapy, and control programs. In this section, we consider an SEIR type model for the COVID-19 disease transmission with two additional compartments of population. One of the compartments consist of asymptomatic individuals while the other one represent the number of death cases at any instant of time. The entire population size $N(t) = N$, is divided into six distinct epidemiological compartments of individuals, namely, susceptible, exposed, asymptomatic, symptomatic, recovered from disease and died due to the disease at any instant of time and their sizes are denoted by $S(t)$, $E(t)$, $A(t)$, $I(t)$, $R(t)$ and $D(t)$, respectively. We assume that individuals enter the population by birth or immigration with a constant recruitment as susceptible (S) and exit by death or by infection as sub-case. We incorporate the social contact matrix in the SEAIRD model via the incidence function to account for contact patterns among various age groups of the population under consideration. The modified incidence functions are given by
$$\lambda_i(t) = \beta \sum_{j=1}^M \bigg(C_{i\,j}^a\dfrac{A_j(t)}{N_j} +  C_{i\,j}^s\dfrac{I_j(t)}{N_j} \bigg),~\mbox{where}~i=1,2,\cdots,M.$$
The modified SEAIRD model equations are
\begin{equation}\label{sys1}
\begin{aligned}
    \dot{S}_i(t) &= \Lambda - \lambda_i(t) S_i(t) - d_0 S_i(t),\\
     \dot{E}_i(t) & = \lambda_i(t) S_i(t) - (\chi + d_1) E_i(t),\\
     \dot{A}_i(t) &= \alpha  \chi E_i(t) -(\gamma_{ar}+\gamma_{as} + d_2) A_i(t),\\
    \dot{I}_i(t) &= (1-\alpha)  \chi E_i(t) + \gamma_{as} A_i(t) -(\gamma_{sr} + \eta + d_3) I_i(t),\\
    \dot{R}_i(t) &= \gamma_{ar} A_i(t)+\gamma_{sr}I_i(t) -d_4 R_i(t),\\
    \dot{D}_i(t) &=\eta I_i(t),
\end{aligned}
\end{equation}
subject to the following initial conditions at time $t=0$:
\begin{align}\label{initial}
    S_i = S_i^0 \geq 0, E_i = E_i ^0 \geq 0, A_i = A_i^0 \geq 0, I_i = I_i^0 \geq 0, R_i = R_i^0 \geq 0, D_i = D_i^0 \geq 0.
\end{align}
The model is schematically depicted by the transmission diagram in Fig. \ref{fig:SEAIRD}, where the incidence due to contact between infected individuals and susceptible population introduces new members in the exposed class with $\beta$ as the transmission rate on contact. Parameter $\Lambda$ denotes the rate of recruitment (birth, immigration) of new members to the susceptible population, whereas $d_j ~ (j=0,1,2,3,4)$ are the natural death rates in the population compartments $S_i, E_i, A_i, I_i, R_i$, respectively. The rates at which the exposed class depletes into $A$ and $I$ classes are $\alpha \chi$ and $(1 - \alpha) \chi$, respectively. Further, the asymptomatic class $A$ transition to symptomatic class $I$ at rate $\gamma_{as}$ and to the recovered class $R$ at rate $\gamma_{ar}$. The symptomatic infectious population either recovers at rate $\gamma_{sr}$ or eventually meet death due to the disease at rate $\eta$. All parameters in the model are assumed to be positive. \\
\tikzstyle{level 1}=[level distance=30mm, sibling distance=30mm]
\tikzstyle{level 2}=[level distance=20mm, sibling distance=20mm]
\tikzstyle{level 3}=[level distance=20mm]

\begin{figure}
\centering
\begin{tikzpicture}[grow=right,->,>=angle 60]
  \node[draw] {$S$}
   	child {node[draw] {$E$}
	child {	child {node[draw]{$I$} }
			child{node[draw]{$A$}} } };
\draw[<-](7.0,-0.75) -- (7.0,0.75);
\draw[->](3,-0.25)--(3, -1);
\draw[->](0,1)--(0, 0.25);
\draw[<-](0,-1)--(0, -0.25);
\draw[->](9.9,0)--(9, 0);
\draw[->](7.3,0.7)--(8.3, 0.3);
\draw[->](7.25,-0.7)--(8.3, -0.3);
\node[draw] at (12.2,0.0) {$D$};
\node[draw] at (10.2,0.0) {$R$};
\draw [->] (7.0,-1.30) to [out=-35,in=-100] (12.2,-0.25);   	
\draw[->](7.3,1.0) to [out=5,in=100] (10.2,0.25);		
\draw[->](7.3,-1.0)  to [out=-5,in=-100] (10.2,-0.25);
\node at (6.7,-0.0) {$\gamma_{as}$};
\node at (9.0,1.3) {$\gamma_{ar}$};
\node at (9.0,-1.3) {$\gamma_{sr}$};
\node at (11.8,-1.4) {$\eta$};
\node at (3.3,-0.75) {$d_1$};
\node at (0,1.3) {$\Lambda$};
\node at (0.3,-0.75) {$d_0$};
\node at (8,0.7) {$d_2$};
\node at (8,-0.7) {$d_3$};
\node at (9.5,0.3) {$d_4$};
\node at (1.4,0.5)[rotate=0] (N) {$\lambda(t)$};
\node at (5.6,0.55)[rotate=27] (N) {$ \alpha\chi$};
\node at (5.7,-0.6)[rotate=-25] (N) {$(1-\alpha)\chi$};
\end{tikzpicture}
\vspace{-0.2cm}
\caption{Disease transmission diagram (SEAIRD Model)} \label{fig:SEAIRD}
\end{figure}
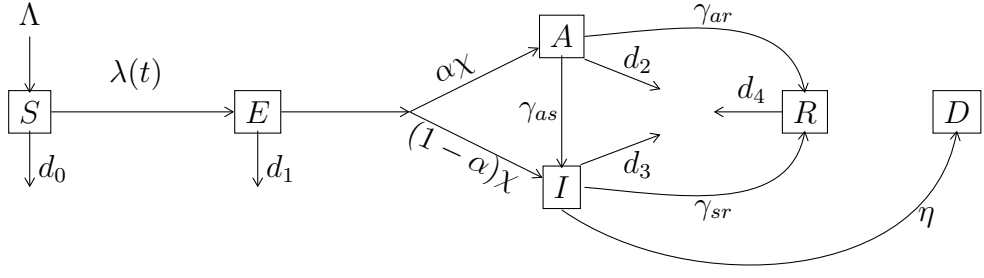
\noindent \emph{Feasible region for the system : } In small time frames, a population may be assumed to be free from demographic changes (birth, death, aging). Without demography our model imply constant population in each age group. That is, $N_i = S_i(t) + E_i(t) + A_i(t) + I_i(t) + R_i(t) + D_i(t)$ is constant for all $t \geq 0$, and $1\leq i \leq M$.  If we assume demography, the model is no longer conservative. Letting $\hat{N}_i(t) = N_i(t) - D_i(t)$, it is easy to see that
$\frac{d \hat{N}_i}{dt} \leq \Lambda - d \hat{N}_i$,
where $d = \mbox{min} \big\lbrace d_0, d_1, d_2, \eta+d_3, d_4  \big\rbrace$. Similarly, from (\ref{sys1}), we obtain
$\displaystyle{\frac{d S_i}{dt}} \leq \Lambda - d_0 S_i.$
Then,
\begin{align*}
    \lim_{t\to\infty} \mbox{sup}\, \hat{N}_i(t) \leq \frac{\Lambda}{d}, ~~ \lim_{t\to\infty} \mbox{sup}\, S_i(t) \leq \frac{\Lambda}{d_0}.
\end{align*}
Hence, feasible region $\Omega_i$ for age-group $i$ may be chosen as the closed set
\begin{align*}
    \Omega_i = \Big \lbrace (S_i, E_i, A_i, I_i, R_i)\in \mathbb{R}^5_+ ~\lvert\, 0\leq S_i \leq \frac{\Lambda}{d_0},\, 0\leq S_i+E_i+A_i+I_i+R_i \leq \frac{\Lambda}{d} \Big\rbrace.
\end{align*}


\subsection{Basic reproduction number and stability}\label{sec:03:stab}
The basic reproduction number or basic reproductive ratio, $R_0$, is defined as the average number of secondary cases generated by a single infectious person in a completely naive population. Here, we introduce the \emph{next-generation-matrix} approach for finding $R_0$ of SEAIRD model. From now on we remove the subscript index $i$ while considering only one age-group. By linearising the dynamics about the disease free equilibrium point $D^{free}_0 = (N,0,0,0,0,0)$, we obtain the transmission matrix, $F$, and the transition matrix, $W$ as follows
$$F=\begin{bmatrix}
0&\beta &\beta \\
0&0&0\\
0&0&0
\end{bmatrix},~~\mbox{ and}~~ W=
\begin{bmatrix}
\chi + d_1 & 0&0\\
-\alpha \chi&\gamma_{ar} + \gamma_{as} + d_2&0\\
-(1-\alpha)\chi &-\gamma_{as}& \gamma_{sr}+\eta + d_2
\end{bmatrix}.$$

$$ FW^{-1} =
\begin{bmatrix}
a_{11} & a_{12} & a_{13}\\
0&0&0\\
0&0&0
\end{bmatrix},$$
where 
\begin{align*}
    a_{11} &= \frac{\beta \alpha \chi(\gamma_{sr}+\eta+d_3)
+ \beta (\gamma_{ar}+d_2)(1-\alpha)\chi+\beta\gamma_{as}\chi}{(\chi+d_1) (\gamma_{ar}+\gamma_{as}+d_2)(\gamma_{sr}+\eta+d_3)},\\
a_{12} &= \frac{\beta (\chi+d_1)(\gamma_{sr}+\eta+d_3)+\beta(\chi+d_1)\gamma_{as}}{(\chi+d_1) (\gamma_{ar}+\gamma_{as}+d_2)(\gamma_{sr}+\eta+d_3)},\\
a_{13} &= \frac{\beta (\chi+d_1)(\gamma_{ar}+\gamma_{as}+d_2) }{(\chi+d_1) (\gamma_{ar}+\gamma_{as}+d_2)(\gamma_{sr}+\eta+d_3)}.
\end{align*}
$R_0$ is the maximum of the absolute eigenvalues of the next generation matrix $FW^{-1}$. Therefore,

\begin{align}\label{R0neq}
    R_0 = \frac{\beta \alpha \chi(\gamma_{sr}+\eta+d_3)
+ \beta (\gamma_{ar}+d_2)(1-\alpha)\chi+\beta\gamma_{as}\chi}{(\chi+d_1) (\gamma_{ar}+\gamma_{as}+d_2)(\gamma_{sr}+\eta+d_3)}.
\end{align}
If $\Lambda = 0, \,d_j = 0~ (j=0,1,2,3,4)$, then
\begin{eqnarray}\label{r0}
R_0 = \displaystyle{\frac{\beta \alpha (\gamma_{sr}+\eta)+ \beta \gamma_{ar}(1-\alpha)+\beta\gamma_{as}}{(\gamma_{ar}+\gamma_{as})(\gamma_{sr}+\eta)} = \dfrac{\beta \alpha}{(\gamma_{ar}+\gamma_{as})} + \dfrac{\beta \gamma_{ar}(1-\alpha)+\beta\gamma_{as}}{(\gamma_{ar}+\gamma_{as})(\gamma_{sr}+\eta)}}.
\end{eqnarray}
\noindent Showing local or the global stability for nonlinear dynamical system such as SIR and SIS type epidemic model were never easy for the researchers. People have tried multiple approach and several years to show the stability of these models and prove them with minimal conditions. The higher dimensional dynamical systems like SEIR and SEIS were known to be globally stable at the disease free equilibrium (DFE) subject to $R_0\leq 1$. Global stability of the endemic equilibrium for these systems were conjectured when $R_0>1$ but remained open for a long time. This conjecture was solved by Li and Muldowney \cite{Myli:1995} in 1995. To establish this, they have used  Poincar\'e Bendixson criterion in three dimensions. Following this, the global stability properties of SEIRS type  has been improved by Fan, Li, Driessche, Wan \cite{Li:1999, fan:2001}. Thereafter, Korobeinikov and Maini \cite{ Kor:2004, Maini:2004} studied the global stability of SEIR and SEIS type models using Lyapunov functions. In this paper, we discuss local stability results for DFE and endemic equilibrium. Recent contributions \cite{shu:2011, supa:2018, liu, side:2016} may be referred for local stability results on epidemic models. The stability analysis will support our mathematical model and the subsequent computational investigations. We now discuss local stability of the disease free equilibrium for the SEAIRD model (\ref{sys1}). For this purpose, we need to drop the demographic parameters (natural birth and death) in the model. 
\begin{theorem}\label{th1}
The disease free equilibrium point $D^{free}_0$ is locally  stable if $R_0<1$.
\end{theorem}
\begin{proof}
To prove local stability of the disease free equilibrium, we need to establish the stability of the system linearised about the equilibrium point. Therefore, it is enough to show that the eigenvalues of the corresponding Jacobian have only non-positive real parts. The Jacobian of the system (\ref{sys1}) at $D^{free}_0$ is given by
$$J(D^{free}_0)=\begin{bmatrix}
0&~~0&-\beta &-\beta &0~&0\\
0&-\chi&~~\beta &~~\beta &0~&0\\
0&~~\alpha \chi&-(\gamma_{ar}+\gamma_{as)}&~~0& 0~&0\\
0&(1-\alpha)\chi&\gamma_{as}&-(\gamma_{sr}+\eta)&0~&0\\
0&~~0&~~\gamma_{ar}&~~\gamma_{sr}&0~&0\\
0&~~0&~~0&~~\eta &0~&0
\end{bmatrix}.$$
Since the rank of the Jacobian matrix is $3$, its characteristic equation has three zero roots. The non-zero roots, also three in number, are obtained as roots of the cubic equation 
\begin{equation}\label{char1}
    \begin{aligned}
    y^3 + a_1 y^2 + a_2 y + a_3 = 0,
    \end{aligned}
\end{equation}
whose coefficients $a_1$, $a_2$, $a_3$ are as follows:
\begin{align}\label{a1}
    a_1 &= \gamma_{ar}+\gamma_{as}+\chi +\gamma_{sr}+\eta \\ \label{a2}
    a_2 &= \chi(\gamma_{ar}+\gamma_{as})+(\eta+\gamma_{sr})(\gamma_{ar}+\gamma_{as}+\chi)-\beta \chi,\\
    a_3 &= \chi(\gamma_{ar}+\gamma_{as})(\gamma_{sr}+\eta)-\beta \gamma_{ar}(1-\alpha)\chi-\beta \gamma_{as}\chi-\beta\alpha\chi(\gamma_{sr}+\eta).  \label{a3}
\end{align}
By Routh-Hurwitz Criterion ( see, \cite{Book:Routh}:1.6-6),  the  roots  of  the  equation (\ref{char1})  have  negative  real  parts  if  and  only  if $a_1>0,~a_3>0,~\mbox{and}~a_1 a_2 > a_3$. From relation (\ref{a1}) it is obvious that $a_1>0$. Relation (\ref{a3}) may be rewritten as 
\begin{align*}
   a_3 &=\chi (\gamma_{ar}+\gamma_{as})(\gamma_{sr}+\eta)\Bigg( 1 - \frac{\beta \alpha (\gamma_{sr}+\eta)+ \beta \gamma_{ar}(1-\alpha)+\beta\gamma_{as}}{(\gamma_{ar}+\gamma_{as})(\gamma_{sr}+\eta)} \Bigg)\\
    &=\chi (\gamma_{ar}+\gamma_{as})(\gamma_{sr}+\eta)(1-R_0) >0,~\mbox{if}~ R_0<1.
\end{align*}
The last condition $a_1 a_2 > a_3$ may be easily verified in some simple steps. Multiplying (\ref{a1}) and (\ref{a2}), we obtain
\begin{align*}
        a_1 a_2 = (\gamma_{ar}+\gamma_{as})(\gamma_{sr}+\eta)(\gamma_{ar}+&\gamma_{as}+\gamma_{sr}+\eta) + a_1\cdot (\gamma_{ar}+\gamma_{as}+\gamma_{sr}+\eta-\beta)\\
    +&a_3+\beta \gamma_{ar}(1-\alpha)\chi+\beta \gamma_{as}\chi+\beta\alpha\chi(\gamma_{sr}+\eta).
\end{align*}
In view of the above relation, we now need to show that the second term is positive. From (\ref{r0}) and the condition $R_0<1$, we obtain
\begin{equation}\label{inq1}
\gamma_{ar}+\gamma_{as} - \beta \alpha >0 \mbox{ and }  \beta \gamma_{ar}(1-\alpha) + \beta\gamma_{as} < (\gamma_{ar}+\gamma_{as})(\gamma_{sr}+\eta),
\end{equation} since $R_0$ is the sum of two positive quantities. Again, with the help of the inequalities in (\ref{inq1}), we obtain the desired inequality 
$$ \gamma_{ar} + \gamma_{as} + \gamma_{sr} + \eta -\beta > 0.$$
\end{proof}
\noindent \emph{Existence of endemic equilibrium point }:
We now discuss the existence of endemic equilibrium for the SEAIRD model (\ref{sys1}) assuming demography. If $R_0>1$, then relation (\ref{R0neq}) result in the following inequality
\begin{equation}\label{R0inq2}
\begin{aligned}
\beta \chi \big(\alpha (\gamma_{sr}+\eta+d_3)+  (1-\alpha)(\gamma_{ar}&+\gamma_{as}+d_2)+
\alpha \gamma_{as}\big) > \\ &(\chi+d_1) (\gamma_{ar}+\gamma_{as}+d_2)(\gamma_{sr}+\eta+d_3).
\end{aligned}
\end{equation}
 To obtain endemic equilibrium point for the SEAIRD model, we set $ S' = 0, E' = 0, A'= 0, I' = 0,\mbox{ and } R' = 0$, so that,
 \begin{align}\label{homog}\nonumber
    0 =& ~\Lambda - \frac{\beta}{N}(A^*+I^*) S^* - d_0 S^*,\\\nonumber
    0 =& ~\frac{\beta}{N}(A^*+I^*)  S^* - (\chi + d_1) E^* ,\\
    0 =&~ \alpha  \chi E^* -(\gamma_{ar}+\gamma_{as} + d_2) A^*,\\\nonumber
   0 =&~  (1-\alpha)  \chi E^* + \gamma_{as} A^* -(\gamma_{sr} + \eta + d_3) I^*, \\\nonumber
   0 =&~ \gamma_{ar} A^*+\gamma_{sr}I^* - d_4 R^*. 
\end{align}
By solving the homogeneous system (\ref{homog}) in terms of $I^*$ we obtain the following 
\begin{align}\label{homog_soln}\nonumber
    S^* &= \frac{\Lambda}{\dfrac{\beta}{N} (\alpha \chi  m  + 1 )I^* + d_0 },  &
    E^* &= (\gamma_{ar}+\gamma_{as} + d_2)mI^*,\\
    A^* &= \alpha \chi m I^*, &
    R^* &= \frac{1}{d_4}(\alpha \chi  \gamma_{ar}  m + \gamma_{sr})I^*,
\end{align}
where $m = \displaystyle{\frac{\gamma_{sr}+\eta + d_3}{P}}$, and $P =  (1-\alpha)  \chi(\gamma_{ar}+\gamma_{as} + d_2) + \alpha  \chi \gamma_{as}$.\\[0.5ex]
Substituting the expression for $S^*$ in the second equation of (\ref{homog}) lead to
\begin{align*}
I^* = \dfrac{\dfrac{\beta}{N}(\alpha \chi m + 1)  \Lambda  - d_0(\chi + d_1)(\gamma_{ar}+\gamma_{as} + d_2)m}{\dfrac{\beta}{N} (\alpha \chi  m  + 1 )(\chi + d_1)(\gamma_{ar}+\gamma_{as} + d_2)m}.
\end{align*}
Non-trivial solution of the homogeneous system (\ref{homog}) exists when $I^*>0$. It is easy to see that denominator part of $I^*$ is positive. To establish the positivity of $I^*$, first we rewrite inequality (\ref{R0inq2}) in the form
\begin{equation}
\beta(\alpha \chi m + 1)   > (\chi+d_1) (\gamma_{ar}+\gamma_{as}+d_2)m. \label{i310}
\end{equation}
Making use of inequality (\ref{i310}) in the expression for $I^*$ leads to
\begin{align*}
  & I^*  > \frac{N}{\beta(\alpha \chi m + 1)} \Big(\frac{\Lambda}{N}-d_0\Big)>0, ~~\mbox{since}~\frac{\Lambda}{N}>d_0.
\end{align*}
Therefore, endemic equilibrium exists when $R_0>1$. Next, we present stability result for the endemic equilibrium point $D^{end}_0$.
\begin{theorem}\label{th2}
The endemic equilibrium point $D^{end}_0 = (S^*, E^*, A^*, I^*, R^*)$ of the SEAIRD model (\ref{sys1}) is asymptotically stable if $R_0>1$.
\end{theorem}
\begin{proof}
As in Theorem \ref{th1}, one may compute the Jacobian at $D^{end}_0$ 
$$J\big{(}D^{end}_0\big{)}=\begin{bmatrix}
-(\beta \frac{A^*}{N}+\beta \frac{I^*}{N} + d_0)&~~0&-\beta \frac{S^*}{N}&-\beta \frac{S^*}{N}&0\\
~~\beta \frac{A^*}{N}+\beta \frac{I^*}{N}&-(\chi+d_1)&~~\beta \frac{S^*}{N}&~~\beta \frac{S^*}{N}&0\\
~~0&~\alpha\chi&-(\gamma_{ar}+\gamma_{as}+d_2)&~~0&0\\
~~0&(1-\alpha)\chi&~~\gamma_{as}&-(\gamma_{sr}+\eta +d_3)&0\\
~~0&~~0&~~\gamma_{ar}&~~\gamma_{sr}&-d_4
\end{bmatrix}.$$
The characteristic equation of the matrix $J\big{(}D^{end}_0)$ is given by
\begin{align}\label{char2}
    \big|  J\big{(}D^{end}_0) - \lambda I \big| = 0.
\end{align}
Simplifying equation (\ref{char2}), one may easily obtain  $\lambda = -d_4$ as an eigenvalue. Other eigenvalues are roots of the equation
\begin{equation}\label{eig}
    \begin{aligned}
\big( \lambda +m_1\big)(\lambda +m_2) (\lambda + m_3) (\lambda +m_4) - \frac{ \beta \chi S^*}{N} (\lambda+d_0)\big( \lambda + \alpha m_4 +  (1-\alpha)m_3+
\alpha \gamma_{as}\big)  = 0
    \end{aligned}
\end{equation}
where
\begin{equation}
\begin{aligned}
    m_1 &=   \frac{\beta A^*}{N} + \frac{\beta I^*}{N} + d_0> 0, &
    m_2 &=  \chi + d_1 > 0,\\
   m_3  &=  \gamma_{ar} + \gamma_{as} +d_2 > 0,  &
   m_4 &=  \gamma_{sr} + \eta + d_3 > 0,
\end{aligned}
\end{equation}
More succinctly, equation (\ref{eig}) may be expressed as  
\begin{equation}\label{l4}
    \begin{aligned}
     \lambda^4 + a_1 \lambda^3 + a_2 \lambda^2 + a_3 \lambda + a_4 = 0,
    \end{aligned}
\end{equation}
with coefficients
\begin{align*}
    a_1 &= m_1+ m_2 + m_3 + m_4 ,\\
    a_2 &= (m_1+m_2)(m_3+m_4) + m_1m_2 + m_3m_4 - \frac{\beta \chi S^*}{N},\\
    a_3 &=  (m_1+m_2)m_3m_4 + (m_3+m_4)m_1m_2 -  \frac{\beta \chi S^*}{N}(\alpha m_4 +  (1-\alpha)m_3+
\alpha \gamma_{as}+d_0),\\
    a_4 &= m_1m_2m_3m_4 -d_0 \frac{\beta \chi S^*}{N} (\alpha m_4 +  (1-\alpha)m_3+
\alpha \gamma_{as}).
\end{align*}
To show that the roots of equation (\ref{l4}) have negative real parts, we use the Routh-Hurwitz Criterion (see, \cite{Book:Routh}:1.6-6) for fourth degree polynomial, which in terms of the coefficients translates to the following conditions
\begin{equation}
a_1>0,~a_3>0,~a_4>0,~a_1a_2a_3>a_3^2+a_1^2a_4. \label{rhc}
\end{equation}
It is obvious that $a_1 >0$. The second equation in (\ref{homog}), in view of (\ref{homog_soln}) and the conditions $R_0>1$, $\beta <1$ lead to 
\begin{equation}\label{S_ter2}
   \frac{\beta }{N} S^* = \frac{(\chi + d_1)E^*}{(A^* + I^*) } = \frac{m_2m_3m_4}{\chi\big(\alpha m_4 + (1-\alpha)m_3 + \alpha \gamma_{as}\big)} =\frac{\beta}{R_0}< 1.
\end{equation}
Making use of (\ref{S_ter2}), expressions of $a_3$, $a_4$ may be rewritten in the following form establishing their positivity
\begin{align*}
a_3 &=  m_1(m_3m_4 + m_2m_3+ m_2m_4) -  d_0 \frac{\beta \chi S^*}{N} > d_0\Big(m_3m_4 + m_2m_3+ m_2m_4 - \frac{\beta \chi S^*}{N}\Big) > 0,\\
a_4 &= m_1m_2m_3m_4 -d_0 m_2m_3m_4 = \big( \frac{\beta A^*}{N} + \frac{\beta I^*}{N} \big)m_2m_3m_4 > 0.
\end{align*}
In the same way, it can be shown that $a_2>0$. Finally, to establish the last condition in the Routh-Hurwitz Criterion (\ref{rhc}), it is enough to verify the following two inequalities :
\begin{align}
    a_1 a_2 a_3 > 2 a_3^2 ~&\Rightarrow~ a_1 a_2 > 2 a_3, \label{cond11} \\
     a_1 a_2 a_3 > 2 a_1^2 a_4 ~&\Rightarrow~ a_2 a_3 > 2 a_1 a_4.\label{cond21}
\end{align}
To prove inequality (\ref{cond11}), we begin with
\begin{align*}
    &a_1 a_2 - 2 a_3 = (m_1+ m_2 + m_3 + m_4) \bigg( (m_1+m_2)(m_3+m_4) + m_1m_2 + m_3m_4 - \frac{\beta \chi S^*}{N} \bigg)-\\ &~~~~~~~~~~~~~~~~ 2 \big( m_1m_3m_4 + (m_3+m_4)m_1m_2\big).
\end{align*}
Following some simple algebraic manipulations, we get
\begin{align*}
    &a_1 a_2 - 2 a_3 =\frac{1}{P_1}\bigg( \Big(m_1^2m_3 + m_1^2m_4 + m_1^2m_2+ m_1m_2^2 +m_1m_3^2 + m_1m_4^2 +   m_3m_4^2 + m_3^2m_4 + \\&~~~~~~~~~~~~~~~~~ m_1m_3m_4  + m_2m_3m_4   \Big)(\alpha m_4 + (1-\alpha)m_3 + \alpha \gamma_{as}) + C_1 \bigg),
\end{align*}
where $P_1 = \alpha m_4 + (1-\alpha)m_3 + \alpha \gamma_{as}$ and
\begin{align*}
    C_1 &=  ( m_2^2m_3+m_2^2m_4 + m_2m_3^2 +  m_2m_4^2 + m_1m_2m_3 + m_1m_2m_4+2m_2m_3m_4 )\times\\&~~~~(\alpha m_4 + (1-\alpha)m_3 + \alpha \gamma_{as}) -  (m_1+ m_2 + m_3 + m_4) m_2m_3m_4\\
    &=\Big(m_1m_2m_3\big( (1-\alpha)m_3 + \alpha \gamma_{as}\big) + m_1m_2m_4(\alpha m_4 + \alpha \gamma_{as}) \Big)+\\&~~~~
    \Big(m_2^2m_3\big( (1-\alpha)m_3 + \alpha \gamma_{as}\big) + m_2^2m_4(\alpha m_4 + \alpha \gamma_{as}) \Big)+\\&~~~~
    \Big(m_2m_3^2\big( (1-\alpha)m_3 + \alpha \gamma_{as}\big) + m_2m_3m_4(\alpha m_4 + \alpha \gamma_{as}) \Big)+\\&~~~~
    \Big(m_2m_4^2(\alpha m_4 + \alpha \gamma_{as}) + m_2m_3m_4\big( (1-\alpha)m_3 + \alpha \gamma_{as}\big) \Big) > 0.
\end{align*}
Therefore, $ a_1 a_2 - 2 a_3 > 0$. Next we rewrite the expression for $a_2a_3-2a_1a_4$
\begin{align*}
a_2 a_3 - 2 a_1 a_4 &= \bigg( (m_1+m_2)(m_3+m_4) + m_1m_2 + m_3m_4 - \frac{\beta \chi S^*}{N} \bigg) \big( m_1m_3m_4 +\\& (m_3+m_4) m_1m_2 \big) - 2 ( m_1+ m_2 + m_3 + m_4) (  m_1m_2m_3m_4 -d_0 m_2m_3m_4 ).
\end{align*}
Performing similar algebraic manipulations, one may verify inequality (\ref{cond21}).
\end{proof}
\noindent In the next subsection, we computationally validate the effect of asymptomatic infections in rapid transmission of COVID-19. The open-source \texttt{python3} inbuilt ODE solver \texttt{odeint} \cite{madhab:github} has been used to  simulate the SEAIRD model along with the social contact data of 16 different age-class available in \cite{age-struct}, for all our computational results.


\subsection{Asymptomatic and symptomatic infectious}\label{sec:03:asym}
A significant proportion of COVID-19 infections are asymptomatic in nature. Various agencies have estimated the proportion of asymptomatic infections differently. On 21 April 2020, World Health Organization (WHO) declared about $80\%$ of the total infected population are asymptomatic. Later Indian Council of Medical Research (ICMR), India, has declared $69\%$ of COVID-19 infections in India are asymptomatic. Asymptomatic individuals do not show any noticeable symptoms but continue to transmit the infection. The basic SEAIRD model is incorporated with the social contact matrix, therefore, disease transmission by asymptomatic population would take place under the usual social contact pattern. Whereas the symptomatic population may be assumed to transmit at most $10\%$ of the social contact pattern due to reduced contact levels. To understand the contributions of symptomatic and asymptomatic individuals in further transmitting the infection to the susceptible population, we distinguish the infected individuals into two sub-classes asymptomatic ($A$) and symptomatic ($I$). First in Fig. \ref{fig:symptomatic},
we consider the model with zero asymptomatic cases, that is, all the infected cases are symptomatic and obtain the corresponding epidemic evolution curve.\\

\noindent Next, we investigate the influence of asymptomatic infections $A$ on the resulting epidemic evolution curve. We assume that the contact patterns of asymptomatic individual age-groups are same as those of the usual social contact matrix i.e. $C^a = C$.
On the other hand, it is customary to believe that contact pattern of symptomatic individuals would be significantly reduced due to ongoing social epidemic prevention campaigns. Since the contact rate of asymptomatic individuals is greater than that those of symptomatic infected, we let $C^s = f_{sa} C^a$, where $f_{sa}$ lies between $0$ and $1$.
\begin{figure}
  	\includegraphics[width=\linewidth]{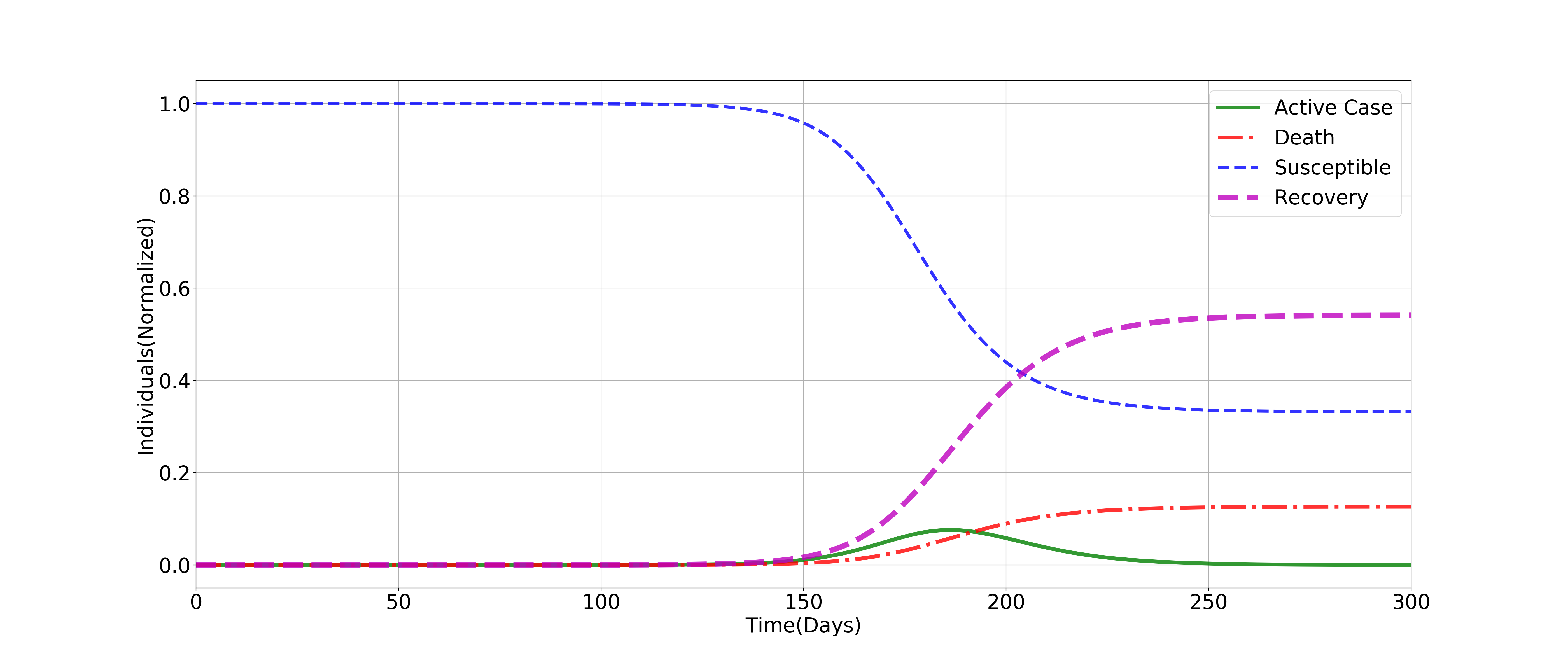}
  	\caption{\textbf{SEAIRD :} Only symptomatic infected class is present without any lockdown ($\alpha = 0.0,~ \beta = 	0.19,~ \chi = 0.29,~ \gamma_{as} = 0.0,~ \eta = 1/30,~ \gamma_{ar} = 0.0,~ \gamma_{sr} = 1/7,~ f_{sa} = 0.1$).}\label{fig:symptomatic}\label{fig:symptomatic}
\end{figure}
Let us assume a situation where $80\%$ infections are asymptomatic and exposed to the susceptible subject to natural social contact pattern. This situation leads to rapid spread of the disease and almost the entire susceptible population gets infected as evident from Fig. \ref{fig:asym}. The susceptible curve quickly falls to zero. Further, the active case peak size is almost three times the $100\%$ symptomatic infections scenario, see Fig. \ref{fig:symptomatic}. A very important point to be observed here is that when asymptomatic cases are more, the number of death cases is significantly reduced. Furthermore, the peak of the active cases is attained around 140 days earlier as compared to the purely symptomatic scenario. This indicates quick disease transmission when asymptomatic infections are more in number. This model is perfect when no control measures are adopted by the policy makers and when the population is very large in size, like 130 crore in India. Then the active case count may hit 40 crore level and about 5\% of the active cases, that is, about 2 crore individuals would turn out to be critical cases needing ICU and ventilator support. Thus, overwhelming the hospital facilities. Therefore, necessary control measures are essential to bring down the peak size of active cases. One of them could be lockdown measures to confine the normal social contact to home contact only. These measures have been adopted by most of the countries till now. Estimating the expected impact of the lockdown, and the potential effectiveness of different exit strategies is critical to inform decision makers on the management of the COVID-19 health crisis.
\begin{figure}
    \includegraphics[width=\textwidth]{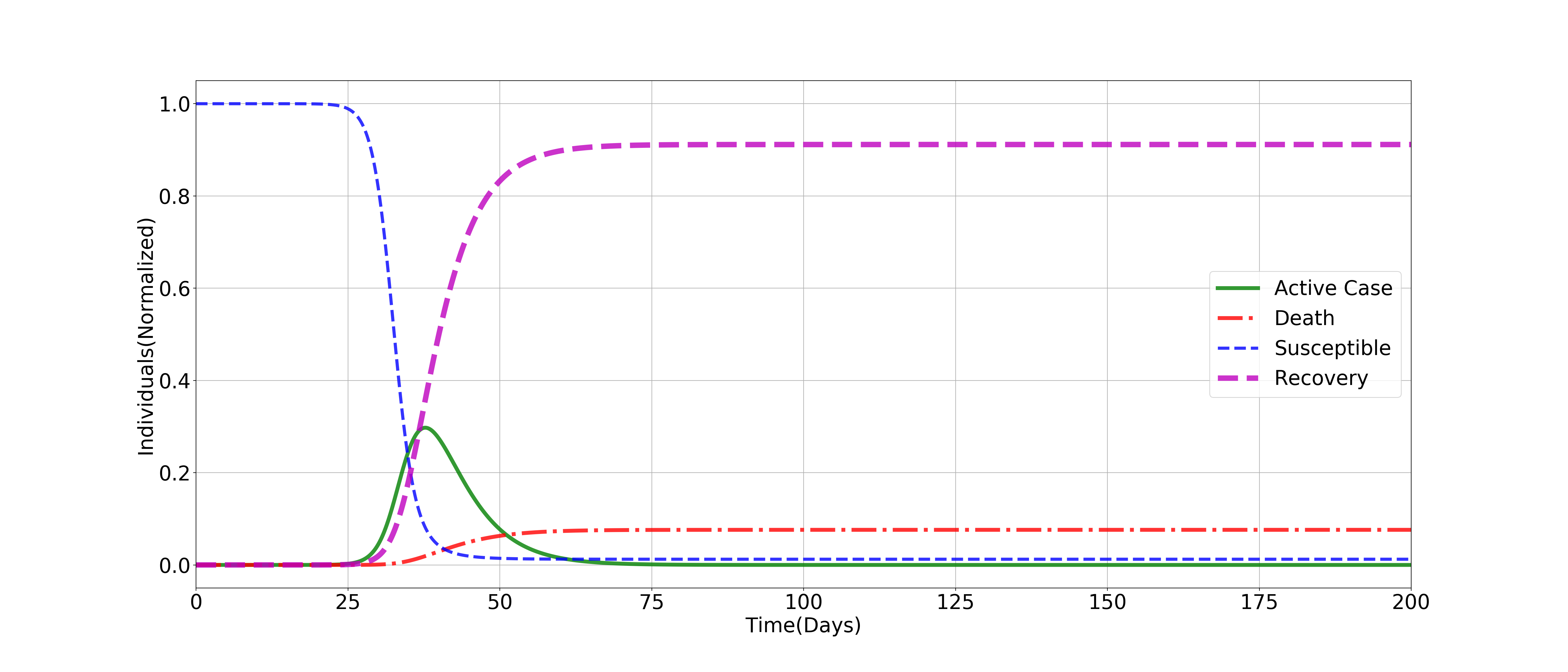}
    \caption{\textbf{SEAIRD :} Both symptomatic and asymptomatic infections are present with no lockdown ($\alpha = 0.8,~ \beta = 0.19,~ \chi = 0.29,~ \gamma_{as} = 0.1,~ \eta = 1/30,~ \gamma_{ar} = 2/7,~ \gamma_{sr} = 1/7,~ f_{sa} = 0.1$)}
    \label{fig:asym}
\end{figure}


\subsection{Impact of lockdown}\label{sec:03:LD}
Rapid transmission and spread of SARS-CoV-2 infections across the globe has led to a situation where more than half of the global population has been put through strict lockdowns and other forms of social distancing measures. More than $90$ countries, including India has been under some form of lockdown simultaneously. This phase has generated global economic turmoil and human miseries. In this context, modelling the impact of various lockdown strategies is of paramount importance. Estimating the expected impact of the lockdown, and the potential effectiveness of different exit strategies is critical to inform decision makers involved in managing the COVID-19 health crisis. While in countries like Italy, Spain, Germany and UK, the peak of the pandemic occurred within weeks of the national lockdown. In India, however, $90$ days have passed since the lockdown was imposed, but the peak has not yet arrived. This indicates that the spread of COVID-19 in India is not exponential so far but it is still growing. We would like to understand whether a prolonged lockdown can just defer the advent of epidemic peak or it would also reduce the peak size significantly and what would be the percentage of reduction? A case study of China reported in \cite{covid:lock:china}, a significant increase in doubling time from 2 days (95\% CI: 1.9 -- 2.6) to 4 days (95\% CI: 3.5 -- 4.3), after imposing lockdown. Researchers also investigated the impact of lockdown in France \cite{covid:lock:france1} and estimated the basic reproduction number at 3.0 (95\% CI: 2.8 -- 3.2) before lockdown and the population infected by COVID-19 as of April 5, 2020 to be in the range 1\% to 6\%. \\

\noindent The average number of contacts is assumed to be reduced by 80\% during lockdown, thereby, leading to a substantial reduction in the reproduction number. But then question arises, can we extend the lockdown longer enough to bring down the number of active case to zero. A study on France COVID-19 data \cite{covid:lock:france2}, based on SEIR model observed that the social distancing is not enough to control the outbreak. We can see in Fig. \ref{fig:lc} that the lockdown can reduce the active case and also the peak may be deferred but full lockdown can't stop the growth of the epidemic. We are trying to see the influence of different lockdown strategies in a 200 days window.  Full lockdown ($100\%$ implementation) disconnect all forms of contacts and only the home contact is allowed, which is imposed for 200 days.
Staggered  easing of lockdown, in Fig. \ref{fig:SLD}, is made phase wise for 93 days only and remaining days are fixed controls by only 20 percent lockdoown, which very marginal and only for the containment zones. As per the lockdown strategy made by Govt. of India  there are five different phases until
June 30, Phase 1 (25 March -- 14 April), Phase 2 (15 April -- 3 May), Phase 3 (May 4 -- May 17), Phase 4 (May 18 -- May 31), and Phase 5 (June 1 -- June 30). These two lockdown strategies compare with $0\%, 60\% \mbox{ and } 80\%$ uniform lockdown for 100 days. \\

\noindent\textcolor{black}{\it Lockdown basis function:} 
\noindent Lockdown implementation in the epidemic models is achieved through the social contact matrix. For this purpose the contact matrix is assumed to be a function of time. We consider a time-dependent control
\begin{align}\label{LDfun}
u(t) = 1 + \frac{P_{ld}}{2} \Big\{ \tanh\Big( \frac{t - t_{off} - t_{woff}}{t_{woff}} \Big) -\tanh\Big( \frac{t - t_{on}}{t_{won}} \Big)\Big\},
\end{align}
where $t_{on}$ and $t_{off}$, respectively, denote lockdown start and end dates. 
Initially, the function value is `1' and gradually decrease and meet at `0' when lockdown is released. The delay in implementation of lockdown is tuned by $t_{won}$ and $t_{woff}$. Effectiveness of lockdown, expressed in percentage, may be controlled by parameter $P_{ld}$. 
\begin{figure}
    \includegraphics[width=\textwidth]{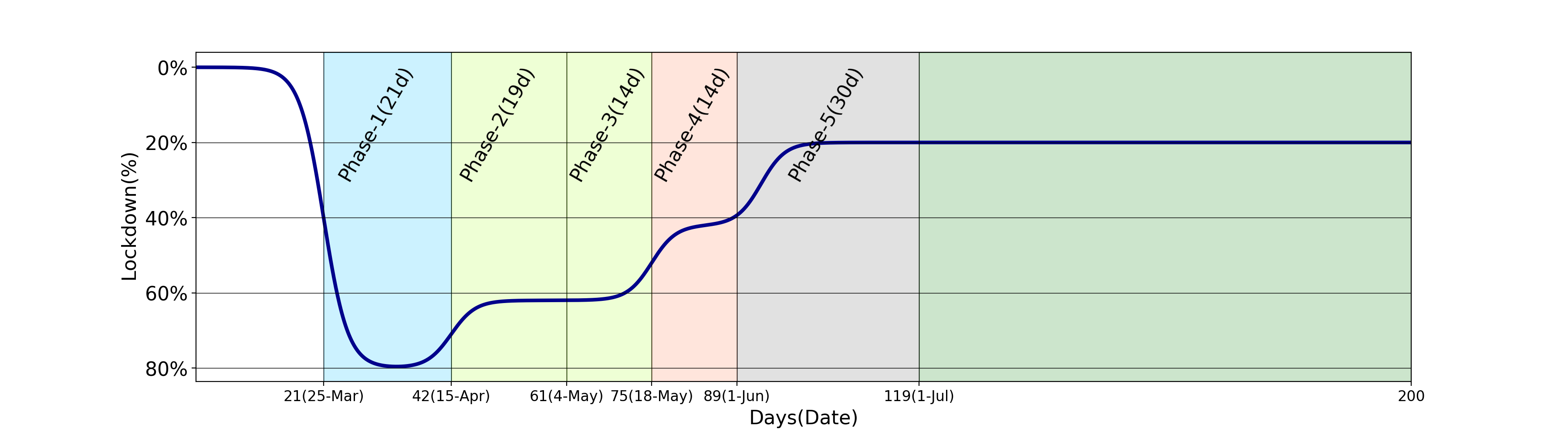}
    \caption{\textbf{Lockdown Strategy:} white color: no lockdown; sky blue: $80\%$ lockdown; greenish yellow: $60\%$ lockdown; coral: $40\%$ lockdown; grey and green: $20\%$ lockdown.}\label{fig:SLD}
\end{figure}
Partitioning contacts into spheres of home, workplace, school and all other categories, the time-dependent contact matrix may be written as\\[-2ex]
\begin{align}
C_{ij}(t) = C^H_{ij} + u^W (t)C^W_{ij} + u^S (t)C^S_{ij} + u^O (t) C^O_{ij}
\end{align}
where $u^W (t)$, $u^S(t)$ and $u^O(t)$ are the corresponding lockdown control functions on the contact matrices of workplace, school and others. Fig. \ref{fig:SLD} is obtained by a suitable linear combination of the control basis functions (\ref{LDfun}).\\
\begin{figure}
    \includegraphics[width=\textwidth]{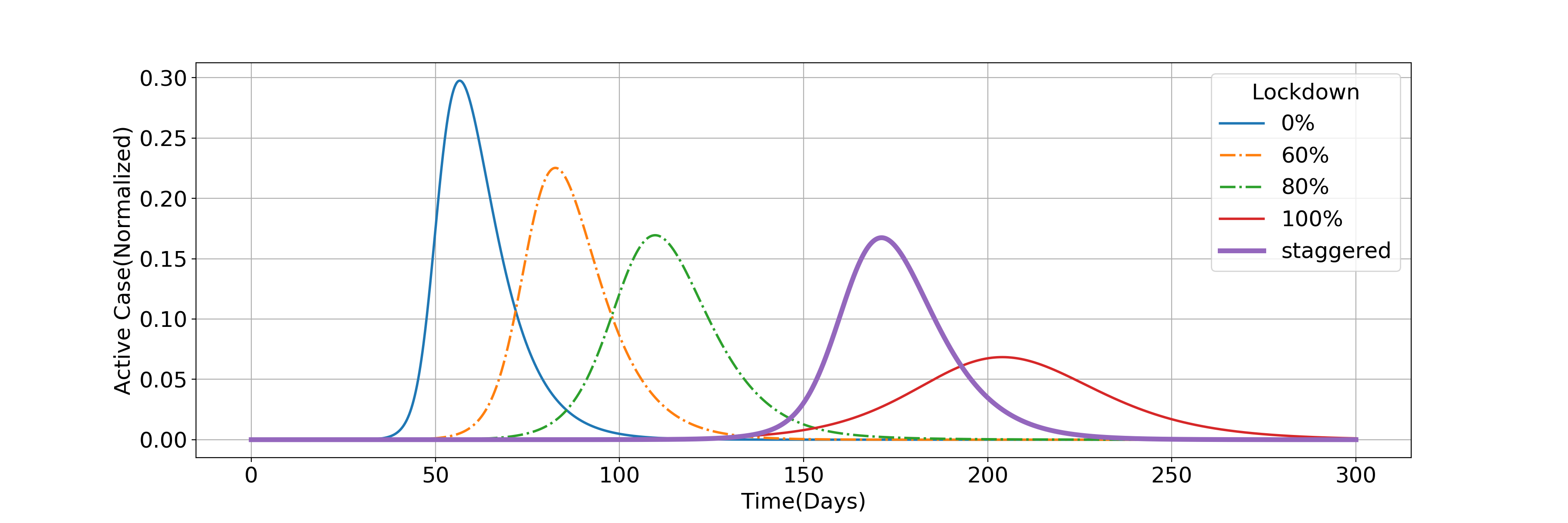}
    \caption{\textbf{SEAIRD :} Effect of different lockdown ($\alpha = 0.8,~ \beta = 0.19,~ \chi = 0.29,~ \gamma_{as} = 0.1,~ \eta = 1/30,~ \gamma_{ar} = 2/7,~ \gamma_{sr} = 1/7,~ f_{sa} = 0.1$)}
    \label{fig:lc}
\end{figure}

\noindent The observation in Fig. \ref{fig:lc} in all cases except the staggered lockdown, can be seen clearly, the peak of the active infections appear within the uniform lockdown period. With increased intensity of lockdown, peak of the infection gets delayed and also the peak size reduces gradually. It is interesting to observe that when the staggered lockdown of $93$ days (see Fig. \ref{fig:SLD}) duration ends, the exponential growth of infections sets in and the peak appears around 118 days. Delay in the advent of peak infection and reduction in its size is significantly better in the staggered lockdown scenario, than that of all other uniform lockdown scenarios except the unrealistic $100\%$ lockdown case. \\ 

\noindent There are still some questions which need to be investigated and answered. For instance, what is the optimal number of days for which the lockdown should be implemented and what is the impact of lockdown for prolonged time periods. Again, the current model does not have any scope for the interventions like the quarantine, hospitalization/isolation and social distancing practices. In the next section, we also analyse the influence of various intervention strategies on the real data which would eventually help us in estimating relevant modelling parameters.


\section{SEAIRD Model with Control Measures}\label{sec:04}
In this section, we extend our SEAIRD model to an SEAIRD-control measure model by incorporating two additional compartments of quarantined (Q) individuals and hospitalized (H) individuals, motivated by the works \cite{quarantine:2020, feng:2007, feng:2007oct, IMSC:2020}. A very important and interesting feature of this new model is the reverse flow from quarantine class (Q) to susceptible class (S) due to the preventive control measure quarantine, see Fig. \ref{fig:M1}. Later in this section we will see the effectiveness of this feature in significantly bringing down the infection levels. Let $\lambda_i(t)$ denote the incidence function of the age group $i$ due to infected individuals from all other age groups. Then
    $$\lambda_i(t) = \beta(t) \sum_{j=1}^M \bigg(C_{i\,j}^a\dfrac{A_j(t)}{N_j} +  C_{i\,j}^s\dfrac{I_j(t)}{N_j} +  C_{i\,j}^h \dfrac{(1-\rho)H_j(t)}{N_j}\bigg),\; \mbox{ where } \rho \in [0, 1].$$
We have already described about $C^a_{i\,j} \mbox{ and } C^s_{i\,j}$ earlier. Here $C^h_{i\,j}$ denote the number of contacts of the hospitalized person in age group $i$  with susceptible individuals of age group $j$. The new control parameter $\rho$ represent the effectiveness of hospitalization or isolation measure applied on symptomatic individuals. The parameter $\rho = 0$, $0 < \rho < 1$, and $\rho = 1$, respectively, describe completely effective, partially effective, and completely ineffective isolation measure.  The epidemic age and contact-structured SEAIRD model with control measures quarantine and isolation (or hospitalization) may be described by the following system of ODEs:
\begin{align}\label{contro_sys}
\dot{S}_i(t) & = -\lambda_i(t) S_i(t) + (1-q)\phi_{qh} Q_i,  	&   \dot{E}_i(t) & = \lambda_i(t) S_i(t) - \chi E_i,\\ \nonumber
\dot{Q}_i(t) & = \alpha_1 \chi E_i - \phi_{qh}Q_i,   	&	\dot{A}_i(t) & = \alpha_2 \chi E_i- (\gamma_{as} + \delta_{ar}  ) A_i,  \\ \nonumber
\dot{I}_i(t) & = \alpha_3 \chi E_i + \gamma_{as} A_i - ( \phi_{sh} + \delta_{sr}) I_i, 	&	\dot{H}_i(t) & = \phi_{sh} I_i + q\phi_{qh}Q_i -(\delta_{hr} + \eta )H_i, \\\nonumber
\dot{R}_i(t) &= \delta_{ar} A_i + \delta_{sr}I_i + \delta_{hr}H_i,   	&	\dot{D}_i(t) &= \eta H_i,~~~ i= 1,2,\cdots,M,
\end{align}

\noindent subject to the initial conditions considered in (\ref{initial}) along with the additional ones $Q_i = Q_i^0 \geq 0$ and $H_i= H_i^0 \geq 0$. As usual $i$ denotes the $i$-th age-group and parameters $\alpha_1$, $\alpha_2$, $\alpha_3$ are non-negative and satisfy the constraint $\alpha_1 + \alpha_2 + \alpha_3 = 1$. For the $i$-th age-group, $Q_i(t)$ and $H_i(t)$, respectively, denote the number of quarantined and hospitalised individuals at time $t$. The transmission and transition parameters are described below:
\begin{align*}
\beta(t) &= \mbox{rate of infection due to contact at time }t,~~~~~~~~~~~~~~~~~~~~~~~~~~~~~~~~~~~~~~\\
\chi &= \mbox{rate of infected from exposed population},\\
\phi_{qh} &= \mbox{rate of transition from  quarantine to hospitalization},\\
\gamma_{as} &= \mbox{rate of transition of asymptomatic population to symptomatic population},\\
\delta_{ar} &= \mbox{rate of recovery of asymptomatic infected individual},\\
\phi_{sh} &= \mbox{rate of transition of symptomatic population to  hospitalization},\\
\delta_{hr} &= \mbox{rate of recovery of hospitalized individual},\\
\eta &= \mbox{death rate of the population},\\
\gamma_{sr} &= \mbox{rate of recovery of symptomatic infected individual},\\
 \alpha_1 &= \mbox{fraction part of $\chi$ from $E_i$ to $Q_i$ class},\\
 \alpha_2 &= \mbox{fraction part of $\chi$ from $E_i$ to $A_i$ class},\\
 \alpha_3 &= \mbox{fraction part of $\chi$ from $E_i$ to $I_i$ class},\\
 q  &= \mbox{control parameter for quarantined class}. 
\end{align*}

\noindent The parameter $q$ describe the strength of the reverse flow feature, as discussed earlier, in the model. Parameter values $q=0$, $q\in (0,1)$, and $q=1$, respectively, denote completely effective, partially effective, and ineffective reverse flow. In our model, there is a latency period or exposed period ($1/\chi$ days) after transmitting the disease from susceptible to potentially infective persons but before these potential infectives gain symptoms and can transmit infection,  which is well discussed in \cite{book:brauer}. We assume that no one gets the disease during the exposed period. The exposed who are not quarantined become infective at the rate $\alpha_2 \chi$ as asymptotic infections ($A$) and at the rate $\alpha_3 \chi$ as symptomatic infections ($I$). The fractional part $ q \phi_qh$ goes to isolated class and rest of the fractional part $ (1-q)\phi_{qh}$ again moves back to the susceptible population, who are quarantined. From the asymptomatic class, some individuals recover at the rate $\delta_{ar}$ and some persons gain symptoms at the rate $\gamma_{as}$. All the symptomatic members are monitored and leaving from the symptomatic class to isolated class ($H$) at the rate $\phi_{sh}$. Finally, all the isolated members either recover ($R$) at the rate $\delta_{hr}$ or eventually meet death due to the disease ($D$) at the rate $\eta$.

\tikzstyle{level 1}=[level distance=30mm, sibling distance=30mm]
\tikzstyle{level 2}=[level distance=30mm, sibling distance=15mm]
\tikzstyle{level 3}=[level distance=20mm]

\tikzstyle{level 1}=[level distance=30mm, sibling distance=30mm]
\tikzstyle{level 2}=[level distance=30mm, sibling distance=25mm]
\tikzstyle{level 3}=[level distance=25mm]

\begin{figure}
\centering
\begin{tikzpicture}[grow=right,->,>=angle 60]
  \node[draw] {$S$}
   	child {node[draw] {$E$}
	child {child {node[draw]{$I$}
	edge from parent node [below]{$\alpha_3 \chi$}}
	child{node[draw]{$A$}edge from parent node [above]{$\alpha_2 \chi$}}}};
\node[draw] at (6.0,-3.5) {$Q$};
\draw [->](6.3,-3.5) -- (8.2,-3.5);
\node[draw] at (8.5,-3.5) {$H$};
\draw[<-](8.5,-3.25) -- (8.5,-1.5);
\draw[<-](8.5,-1) -- (8.5,1);
\node[draw] at (12.4,0.1) {$D$};
\node[draw] at (10.2,0.1) {$R$};
\draw [->] (8.8,-3.5) to [out=5,in=-90] (12.4,-0.15);   
\draw [->] (8.8,-3.5) to [out=5,in=-90] (10.4,-0.15);   
\draw[->](8.8,1.3) to [out=-12,in=115] (10.2,0.35);		
\draw[->](8.75,-1.3)  to [out=10,in=-100] (10.2,-0.15);
\draw[<-] (0.24,-0.05) to [out=-10,in=190] (5.7,-3.5);	
\draw[->] (3.0,-0.25) to [out=-25,in=150] (6.0,-3.17);	
\node at (4.5,0.25) {$(\alpha_2 + \alpha_3 )\chi $};
\node at (8.8,-0.0) {$\gamma_{as}$};
\node at (8.8,-2.3) {$\phi_{sh}$};
\node at (9.9,1.15) {$\delta_{ar}$};
\node at (10.0,-1.1) {$\delta_{sr}$};
\node at (10.5,-2.3) {$\delta_{hr}$};
\node at (11.5,-2.6) {$\eta$};
\node at (4.9,-1.7) {$\alpha_1 \chi$};
\node at (7.4,-3.7) {$q\phi_{qh}$};
\node at (2.8,-2.25)[rotate=-50] (N) {$(1-q)\phi_{qh}$};
\node at (1.4,0.3)[rotate=0] (N) {$\lambda(t)$};
\end{tikzpicture}
\vspace{-0.2cm}
\caption{Disease transmission diagram (SEAIRD-Control model)} \label{fig:M1}
\end{figure}
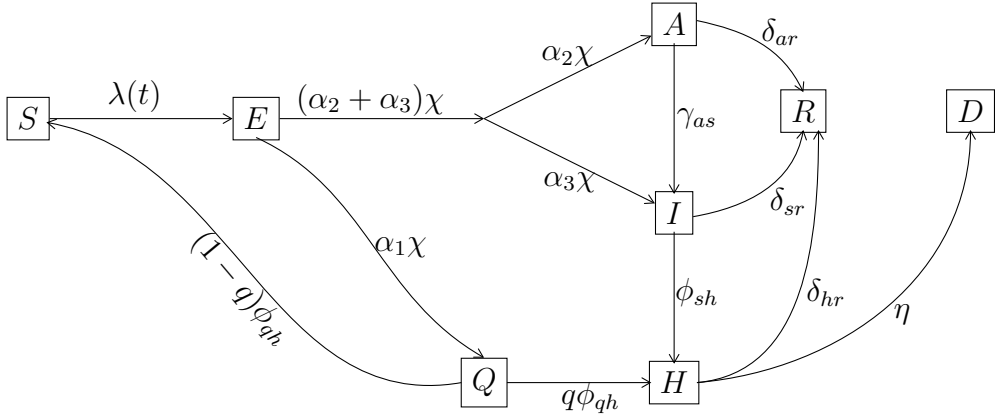


\subsection{Controlled reproduction number}\label{sec:04:R0}
To find the control reproduction number, we  linearise the dynamical system (\ref{contro_sys}) at a disease free equilibrium point $(S_{i}^0, \textbf{0}_{1\times 7} )$, where  $S_{i}^0 = N_i, ~(1\leq i \leq M)$.
Let $Y = [E, Q, A, I, H]^T $ where each of the compartments $E$, $Q$, $A$, $I$, $H$ represent $M$ dimensional vectors. From the linearised system (\ref{contro_sys}), we obtain
\begin{equation}\label{e1}
\dot{Y} = (F - W)Y,
\end{equation}
\noindent where $F$ and $W$, respectively, denote transmission (new infections) and transition (compartmental change). Therefore, expressions for $F$ and $W$ are
$$F = \begin{bmatrix}
0&0& \beta & \beta f_{sa} &f_{sh} \beta (1-\rho)\\
0&0&0&0&0\\
0&0&0&0&0\\
0&0&0&0&0\\
0&0&0&0&0
\end{bmatrix}\otimes K,$$
$$W = \begin{bmatrix}
\chi&0&0&0&0\\
-\alpha_1 \chi&\phi_{qh}&0&0&0\\
-\alpha_2 \chi&0&\gamma_{as}+\delta_{ar}&0&0\\
-\alpha_3 \chi&0&-\gamma_{as}&\phi_{sh} + \delta_{sr} &0\\
0&- q\phi_{qh}&0&-\phi_{sh}&\delta_{hr} + \eta\\
\end{bmatrix} \otimes \textbf{I}_M, $$ where $\otimes$ is the kronecker product and $\displaystyle{K_{i\,j} = \frac{C_{i\,j} N_{i}}{N_j}, \,(1 \leq i, \,j \leq M)}$. Let $R_c$ denote control reproduction number as it is guided by the disease control parameters $\rho$ and $q$ corresponding to the isolated and quarantined populations. Then $R_c$ is defined as the maximum of absolute eigenvalues of the next generation matrix $F W^{-1}$ i.e. $\displaystyle{\mathit{R}_c = \rho (F W^{-1})}$. In more realistic model, we need to consider the rate of new infection $\beta$ as time dependent, which will be discussed in more detail later in the social distancing subsection.\\ 

\noindent\textcolor{black}{\it Effective control reproduction number :} To find the time dependent effective reproduction number $R_c^e(t)$, we replace $N_i$ by $S_i(t)$ and $C_{i,\,j}$ by $C_{i,\,j}(t)$ in the linearised system. Similar to $R_c$, we compute the effective control reproduction number $$\displaystyle{\mathit{R}_c^e(t) = \rho \big{(}F W^{-1}(t)\big{)}}$$ at any time $t$. 
One may refer Feng {\it et al.} \cite{feng:2007} for control reproduction number of the control model and how it is different from basic reproduction number. Again Feng \cite{feng:2007oct} has discussed the exponential and gamma distribution models of latent and infectious periods due to Quarantine ($Q$) and Hospitalization ($H$) measures.


\subsection{Quarantine and hospitalization as new interventions} \label{sec:04:HQ}
In this subsection, we computationally analyse the impact of the control measures on the epidemic spread. 
The contact matrix in the model is of order $M$($= 16$), the number of age groups. For each age group, there are $8$ coupled Ordinary Differential Equations (ODEs) in the SEAIRD-control model. Therefore, a total of $8M$($= 128$) ODEs are required to be solved simultaneously. The open-source \cite{madhab:github} \texttt{python3} inbuilt ODE solver \texttt{odeint} has been used to  simulate the control based epidemic model. Our solver is a modified and extended version of the python based open-source \texttt{pyross} \cite{pyross:github}.\\
\begin{figure}
    \includegraphics[width=\textwidth]{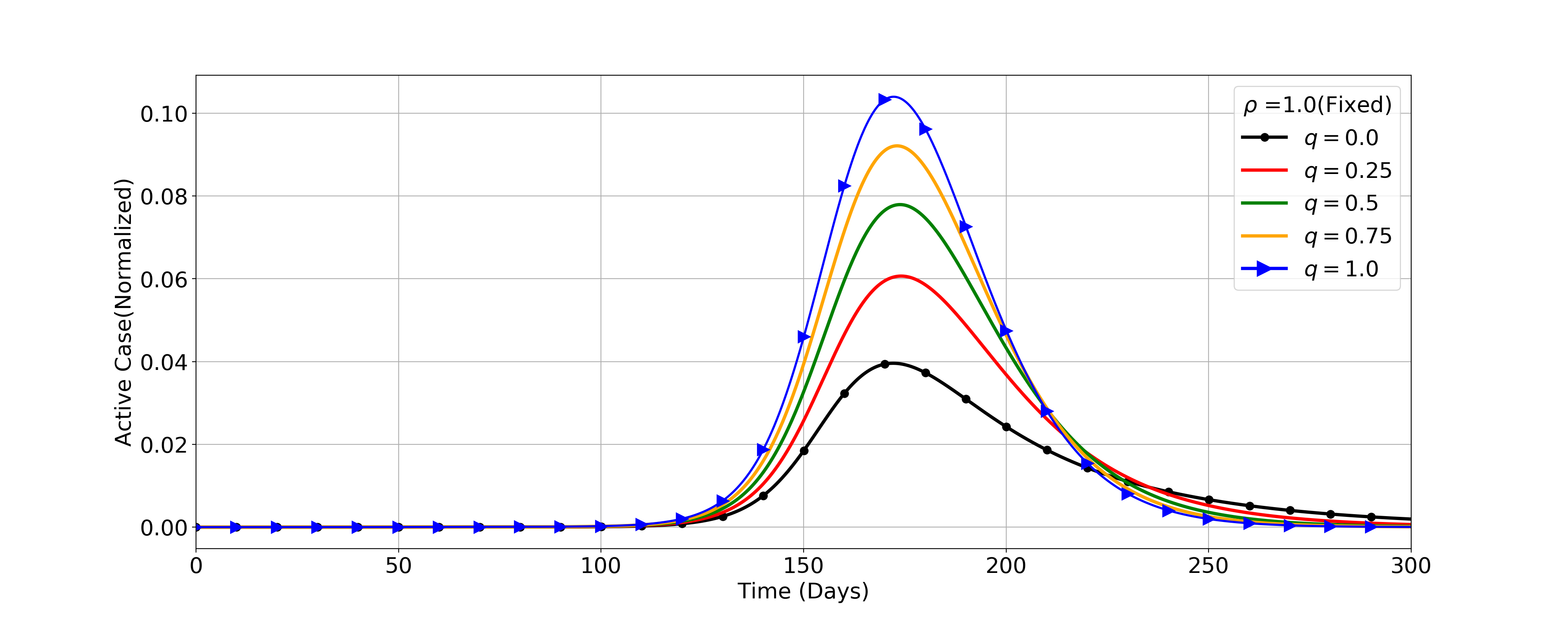}
    \caption{\textbf{SEAIRD-Control :} Figure shows active case profiles for different values of quarantine control parameter $q$ and fixed hospitalization control parameter, $\rho=1.0$. The profiles are reducing in size when $q$ value changes from $1.0$ to $0.0$. The other parameters are as follows: $\alpha_d = 0.05,~ \beta =0.37,~ \chi = 0.29,~ \alpha_1 = 0.7,~ \alpha_2 = 0.2,~ \alpha_3 = 1 - (\alpha_1 + \alpha_2),~ \phi_{qh} = 1/10,~ \gamma_{as} = 0.1,~ \delta_{ar} = 2/7,~ \phi_{sh} = 1/2,~ \delta_{sr} = 1/7,~ \delta_{hr} = (1-\alpha_d)/10,~ \eta = \alpha_d/10,~ f_{sa} = 0.1,~ f_{sh} = 0.1 $ under lockdown strategy of Fig. \ref{fig:SLD}.}
    \label{fig:q_var}
\end{figure}

\noindent The SEIR type mathematical models in epidemiology have been studied in the past with control policies and validated with the real data. Those studies helped in choosing the control parameters appropriately. Many SEIR type control models have been studied to measure the effectiveness of different types of control parameters. In 2003 Lipsitch {\it et al.} \cite{m:2003} studied the control of Severe Acute Respiratory Syndrome (SARS) spread from 2002 to 2003, by taking into account two interventions: ($i$) Isolation of symptomatic cases to prevent further transmission and ($ii$) Quarantine and monitoring of exposed contacts of active cases, so that possible new infections may be identified and isolated easily. \\
\begin{figure}
    \includegraphics[width=\textwidth]{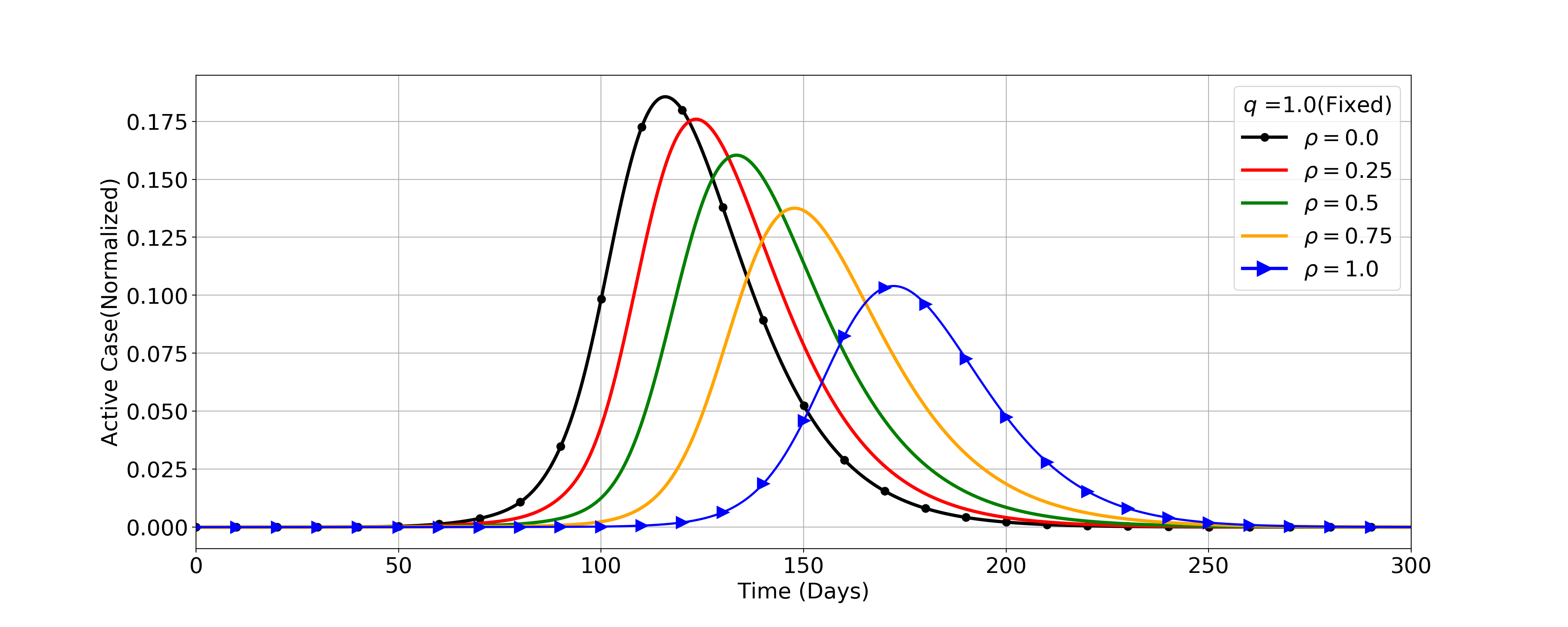}
    \caption{\textbf{SEAIRD-Control :} Figure shows active case profiles for different values of hospitalization control parameter $\rho$ and fixed quarantine control parameter, $q=1$. The profiles shift from left to right as well as reduce in size when $\rho$ value changes from $0$ to $1.0$. The other parameters are same as specified in Fig. \ref{fig:q_var}.}
    \label{fig:rho_var}
\end{figure}

\noindent Fig. \ref{fig:q_var} depict the influence of Quarantine (Q) measure on the epidemic spread when hospitalization is assumed to be completely effective, that is, $\rho=1$. As $q$ value changes from $1$ (ineffective quarantine) to $0$ (completely effective quarantine), the active infections peak size accordingly reduce. Next we fix the quarantine parameter at $q=1$ and simulate the control model (\ref{contro_sys}) while varying $\rho$ value from $0$ to $1$. Fig. \ref{fig:rho_var} display the active infections profiles for different values of $\rho$ ($0.0$, $0.25$, $0.5$, $0.75$, $1.0$). One may see that the profiles shift from left to right, indicating delay in peak arrival, along with peak size reduction. A combination of this two control measures may lead to significant reduction in active infections peak size along with delay in peak arrival. 


\subsection{Social distancing}\label{sec:04:SD}
Here we consider social distancing as a behavioural change in the general population, which need some amount of time to successfully percolate in a society. As the epidemic progress with increased intensity of infection, awareness level for social distancing norms among the population also rises.
Due to this behavioural change new infections rate starts falling. To capture the effect of this behavioural change, we need to introduce time dependent incidence parameter $\beta(t)$ gradually decreasing with time. We consider the logistic function as a model for time dependent incidence parameter $\beta(t)$ and refer it as social distancing function
\begin{equation}\label{sdf}
    \beta(t) = \beta_{min} + \dfrac{\beta_{max} -\beta_{min}}{1 + e^{-k(t - t_m)}}.
\end{equation}
The parameters in (\ref{sdf}) are explained in the caption of Fig. \ref{fig:testA}. We intend to model the impact of social distancing practices along with lockdown policies. During strict lockdown phase (80\%), we assume that the impact of social distancing practices are not very significant because the prevailing social contact pattern is essentially confined to home. But as the lockdowns are gradually relaxed ($< 80\%$), other forms of social contact increase. This is when the impact of social distancing awareness and practices comes into play. Therefore, we introduce social distancing function $\beta(t)$ in our model during 60\% and subsequent lockdown phases with suitable time parameter $t_m$.  \\
 \begin{SCfigure}
 \includegraphics[height=2.0in,width=2.8in]{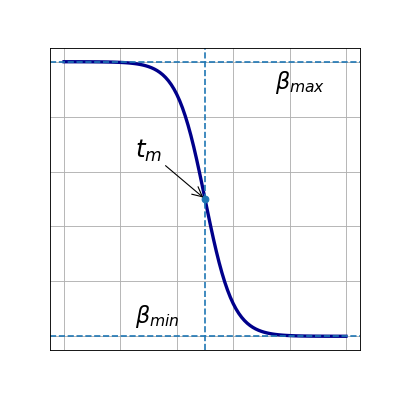}
\hspace{-0.0cm}
\vspace{-0.6cm}
\caption{Here in the social distancing function $\beta_{max}$ = maximum value of $\beta(t)$, $\beta_{min}$ = minimum value of $\beta(t)$, $k$ = steepness of the function or growth rate of the function, and $t_m$ = the midpoint of the sigmoid. This can reduce the contact matrix or the rate of infection by $\beta_{max}-\beta_{min}$ fractions.}\label{fig:testA}
\end{SCfigure}\label{fig:SD}

\noindent We have seen the impact of the control measures quarantine and hospitalization in subsection \ref{sec:04:HQ}. Their combined effect helps in further reducing the number of infections and delaying the advent of its peak significantly. Thus, allowing civic and health authorities with minimum time to ramp up the infrastructure required for better managing the ensuing epidemic outbreak. In addition to these control measures, if a large class of the population develops awareness for social distancing and personal hygiene, the disease transmission can be slowed down further. In the next subsection, we simulate the SEAIRD control model with social distancing practices.


\subsection{Analysis based on real data}\label{sec:04:matching}
In this subsection, we numerically simulate (using Python based solver) the SEAIRD control model (\ref{contro_sys}) laced with the social distancing function (\ref{sdf}) and match the computed results with real COVID-19 data of India till May 15, 2020. By matching the computed results with the real data, we intend to estimate the model parameters to the best possible extent.
\begin{figure}[!h]
    \includegraphics[width=1.0\textwidth]{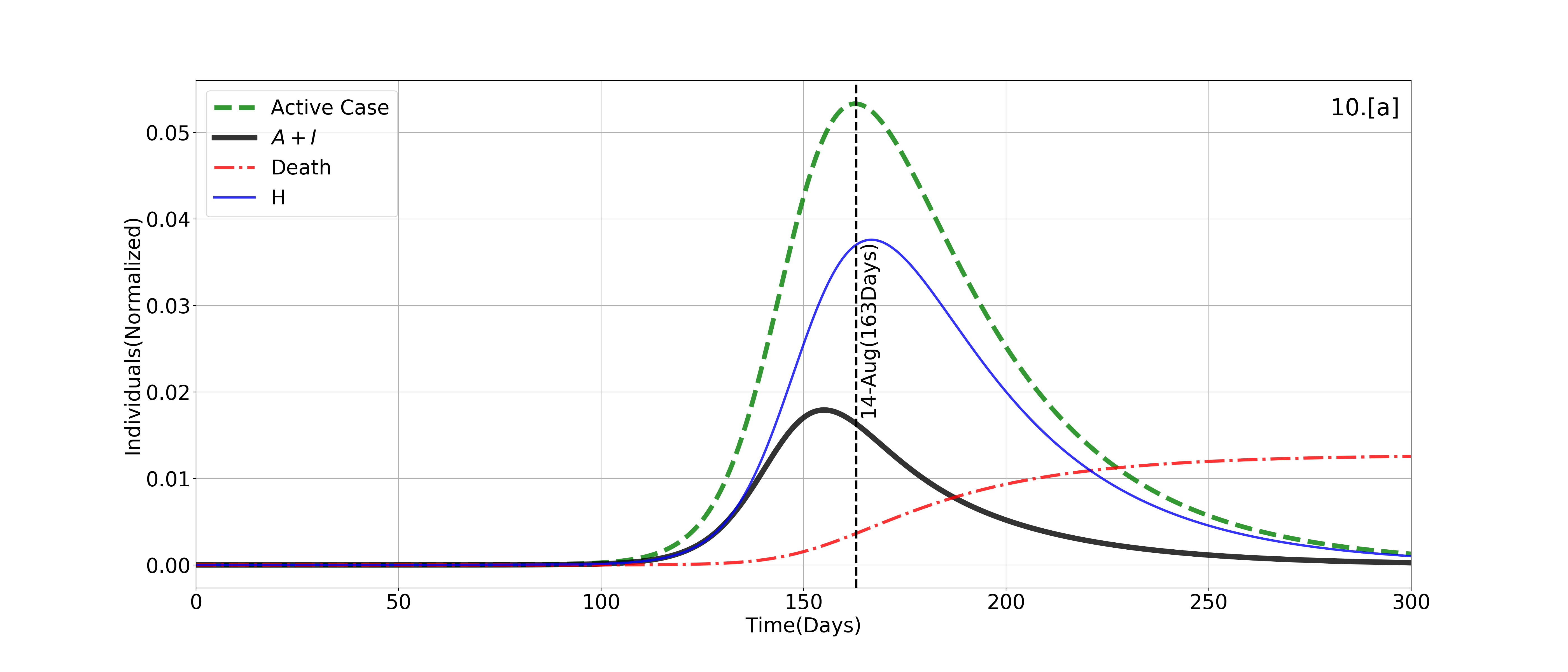}
    \includegraphics[width=1.0\textwidth]{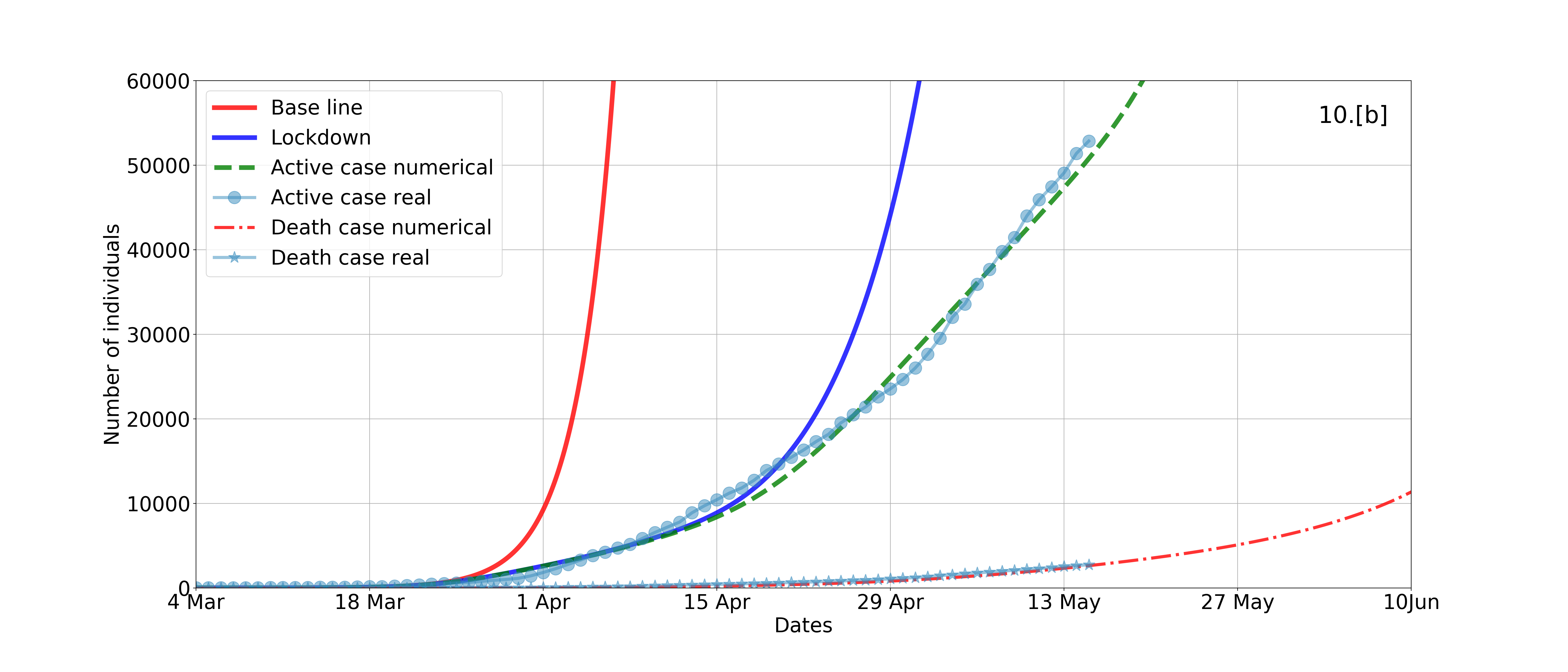}
    \caption{\textbf{SEAIRD-Control :} Subplot [a] depicts active infections under lockdown strategy proposed in Fig. \ref{fig:SLD} with social distancing parameters: $\beta_{max}=0.37,~\beta_{min}=0.21,~t_m=49,~k =0.2$. Model parameters: $\alpha_d = 0.05,~  \rho =0.75,~ \chi = 0.29,~ \alpha_1 = 0.7,~ \alpha_2 = 0.2,~ \alpha_3 = 1 - (\alpha_1 + \alpha_2),~ \phi_{qh} = 1/10,~ q = 0.1,~ \gamma_{as} = 0.1,~ \delta_{ar} = 2/7,~ \phi_{sh} = 1/2,~ \delta_{sr} = 1/7,~ \delta_{hr} = (1-\alpha_d)/10,~ \eta = \alpha_d/10,~ f_{sa} = 0.1,~ f_{sh} = 0.1$. Subplot [b] depicts real data matching with simulated data. \emph{\bf Base line} : active case without social distancing and no lockdown; \emph{\bf Lockdown} : active case without social distancing but with lockdown. Remaining two trajectories (active and death cases simulated) are with social distancing and lockdown matching with corresponding real data.}
    \label{fig:simulation}
\end{figure}
The lockdown policy in Fig. \ref{fig:SLD} is used to convert the contact matrix as a matrix function of time variable. The initial time($t = 0$) of the model (\ref{contro_sys}) is May 04, 2020, that begins with non-zero infection, and the computation is carried out for 300 days. There are three important rapidly growing daily counts named as active cases, symptomatic and asymptomatic total, and hospitalized individuals and steadily increasing number of deaths due to the disease, see Fig. \ref{fig:simulation}.[a]. The active case count is seen to attain its peak around day 163 (August 14, 2020) and about 5 \% of the entire population is infected by the peak day. Furthermore, 1 \% of the population is estimated to meet death due to the disease this year. Finally, the epidemic is seen to wither out after 300 days, diminishing gradually in 135 days. The epidemic spreads with exponential growth for the first 50 days. Fig. \ref{fig:simulation} is again reproduced in a smaller window to show the match between the numerically computed data and the real data up to May 15, 2020.\\

\noindent In Fig. \ref{fig:simulation}.[b], we have plotted the numerically estimated growth of active cases in absence of both lockdown and social distancing function (`Base line'). The `Base line' is seen to match with the real data only up to first 21 days before it gains its exponential growth. Next, we have estimated the active case curve (`Lockdown') under the influence of lockdown intervention. This curve is in agreement with the real data up to the first 45 days and then exponential growth has been observed. In another simulation run, along with lockdown measure behavioural change due to social distancing practices is also accounted for with $t_m = 49$ days. Effect of this behavioural change is reflected in the numerically estimated active case profile which agrees very closely with the real data till May 15, 2020. Furthermore, the numerically computed death curve is also seen to completely match with its real counterpart. \\

\noindent One may like to understand and compare the extent of reduction in infection peak size and delay in its arrival for the two models proposed in Sections \ref{sec:03} and \ref{sec:04}. If we carefully observe Fig. \ref{fig:lc} plotted for the control free SEAIRD model (\ref{sys1}), the staggered lockdown peak at nearly 17\% is seen to be lowest among all other peaks (except the 100\% lockdown case) and also the most delayed peak that appears around 117 days. With the control measures model, the active case peak is reduced by nearly 12 \% followed by an additional delay of 50 days in the peak arrival, see Fig. \ref{fig:simulation}.  


\section{Lockdown Exit Strategy \& Subsequent Waves}\label{sec:05}
There are 29 states in India governed by democratically elected state governments. The spread and impact of COVID-19 on each state has been radically different due to various factors such as exposure to immigration, healthcare facilities, population density, urbanization, and the local government's policies and strategies for the pandemic management and control. Contact tracing and isolation \cite{contact_tracing} is one of the most efficient intervention strategies to significantly control the infection transmission. However, this strategy is not being implemented by the states uniformly. It can be efficiently implemented on relatively small and moderately dense populations like Singapore, Hong Kong and South Korea. For a large and densely populated country like India it may be possible to implement this strategy locally on a small scale due to relatively lesser number of tests per million of population and shortage of man power in administration of medical services \cite{LD:science:2020}.  Our model is not designed to study the contact tracing strategy explicitly but it has been considered as part of our lockdown and social distancing strategy. Similarly, border and travel restrictions are also part of the lockdown and social distancing strategy. Accordingly the modeler would have the flexibility to choose the extent of permissible social contacts during a certain time period in the entire course of the epidemic.\\

\begin{figure}[!h]
  \includegraphics[width=\textwidth]{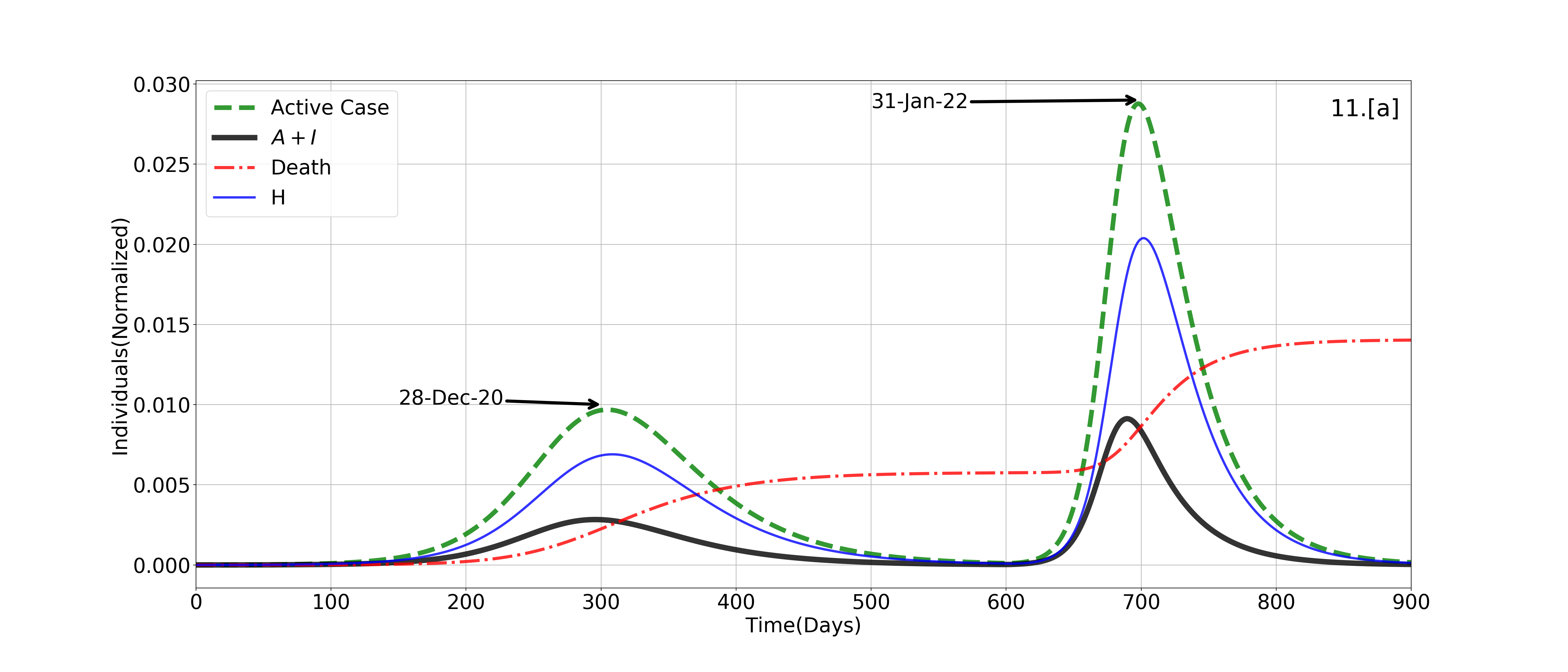}
   \includegraphics[width=\textwidth]{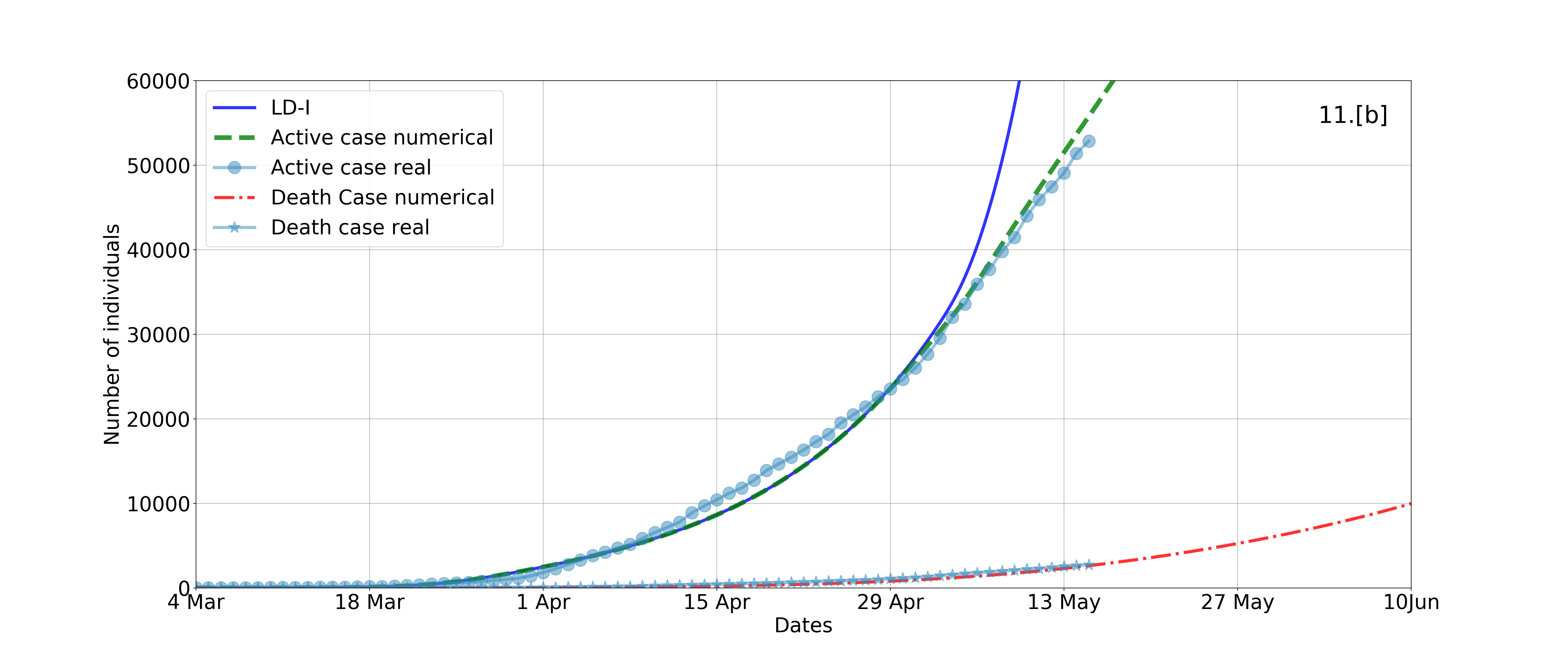}
    \caption{\textbf{SEAIRD-Control :}  Subplot [a] depicts model simulation results under LD-I policy as in Table \ref{table:peak} with social distancing parameters $\beta_{max}=0.379,~ \beta_{min} = 0.21, ~t_m=61,~ k = 0.5$. Model parameters are same as in Fig. \ref{fig:simulation}. Subplot [b] display the agreement between real data and simulation results of Subplot [a]. Also display the simulated active case trajectory ({\bf LD-I}) where social distancing is not considered.}\label{fig:2wave}
\end{figure}
\noindent Gilbert {\it et al.} \cite{LD:NatureMed:2020} have described three different approaches to the exit plans of lockdown which primarily emphasise on a continuous process of intense testing followed by de-confinement of population regions where herd immunity has been attained. Also suggest mathematical models are crucial to ensure that the proposed set of actions would be safe, to make certain that the level of transmission and severe cases remain below the healthcare system?s capacity. Some countries, like UK and Sweden, have envisaged that early onset of herd immunity might be a good way to stop or control the spread of the novel coronavirus. There are many reasons \cite{Herd_Immune:2020} to argue why herd immunity approach is not very efficient in stopping or slowing down the spread of SARS-CoV-2 transmission, rather, entailing into potentially  uncontrollable situation and significant increase in effective reproduction number. Gradually releasing and re-establishing lockdown (when active infections become too high) type on-off exit strategy is discussed in \cite{LD:UK_2020}. An analysis carried on UK COVID-19 data based on SEIR epidemic model reveals that on-off type exit strategy helps in containing the critical cases below the available hospital facilities.\\

\begin{figure}
    \includegraphics[width=\textwidth]{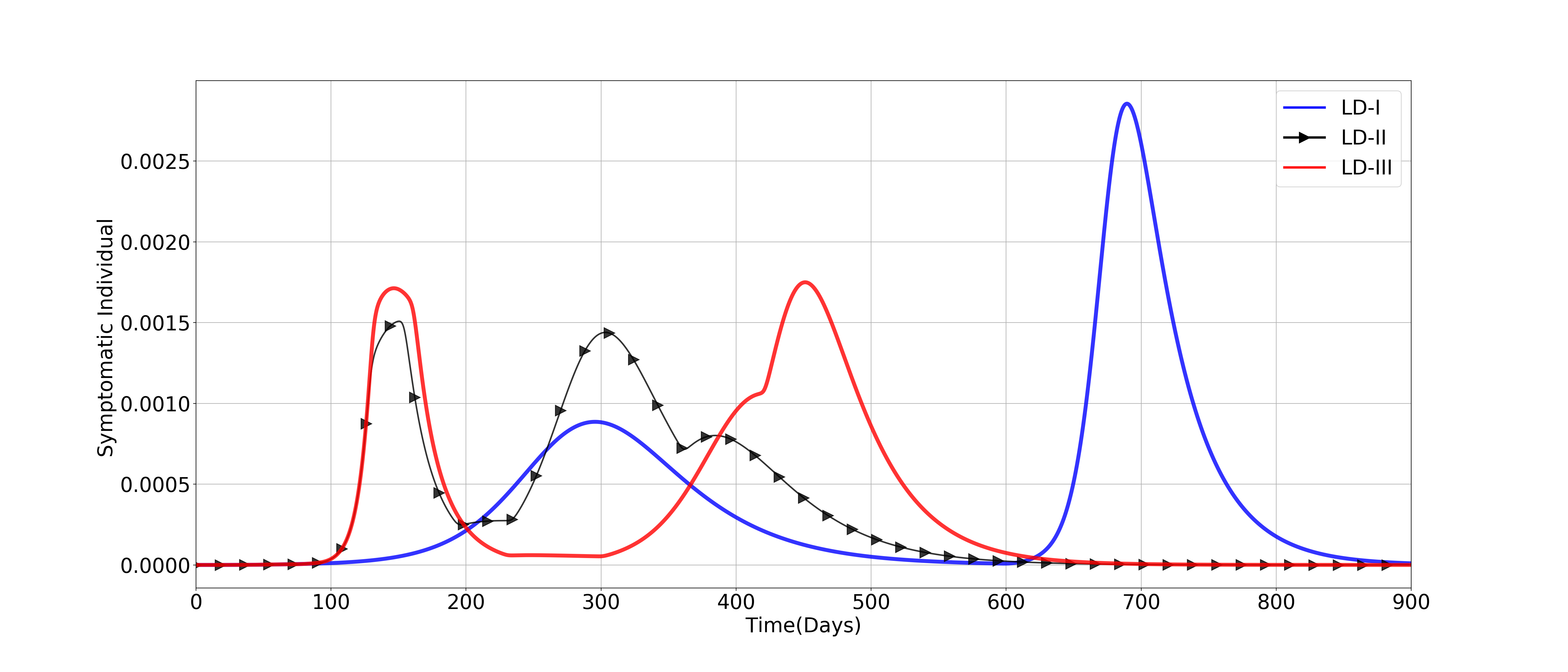}
    \caption{\textbf{SEAIRD-Control :} The growth in symptomatic individuals are plotted for three different lockdown policies as given in Table. \ref{table:peak} and then quantitatively explain in Table. \ref{table:symptomatic}}\label{fig:symwave}
\end{figure}
\begin{table}[h]
\begin{center}
\caption{Three different Lockdown (LD) policies for 900 days.}\label{table:peak}
\begin{tabular}{ |p{1.5cm}|p{13.5cm}| }
\hline
\multicolumn{2}{|c|}{\bf Lockdown (LD) Policies } \\
\hline\hline
\noindent{\bf LD-I} & 80\%(21-60) \textemdash \; \colorbox{orange}{Averaged: 60\%(61-600)} \textemdash \; 20\%(601-900)\\
\hline\hline
\noindent {\bf LD-II} & \colorbox{lightgray}{80\%(21-41) \textemdash \; 60\%(42-74) \textemdash \; 40\%(75-92)} \textemdash \; 20\%(93-127) \textemdash \; 60\%(128-152) \textemdash \; 80\%(153-192) \textemdash \; 60\%(193-232) \textemdash \; 40\%(233-359) \textemdash \; 20\%(360-900)\\
\hline\hline
\noindent {\bf LD-III} & \colorbox{lightgray}{80\%(21-41) \textemdash \; 60\%(42-74) \textemdash \; 40\%(75-92)} \textemdash \; 20\%(93-129) \textemdash \; 60\%(130-159) \textemdash \; 80\%(160-229) \textemdash \; 60\%(230-299) \textemdash \; 40\%(300-419) \textemdash \; 20\%(420-900) \\
 \hline
\end{tabular}
\end{center}
\end{table}
\noindent Table \ref{table:peak} presents three distinct lockdown policies, namely, LD-I, LD-II, and LD-III over a period of 900 days. Policy LD-I mean 80\% lockdown implementation for 40 days starting from day 21 since the first reported infection (04 Mar 2020). The assumption of 80\% lockdown is based on the fact that essential services remain functional. After the first phase of lockdown, whole of India was categorised into either of three zones -- Red, Orange, or Green -- depending on the intensity of infection. Red zones (80\% lockdown) correspond to containment zones; Orange zones (60\% lockdown) for moderately infected zones; Green zones (40\% lockdown) for zones below some threshold infection level. To represent the overall intensity of nationwide lockdown in India, we have considered the arithmetic mean of zonal lockdown intensities, which is 60\%.  Thus, in LD-I policy, first phase of strict lockdown is followed by 60\% and 20\% subsequent lockdown phases of 540 and 300 days, respectively.\\
        
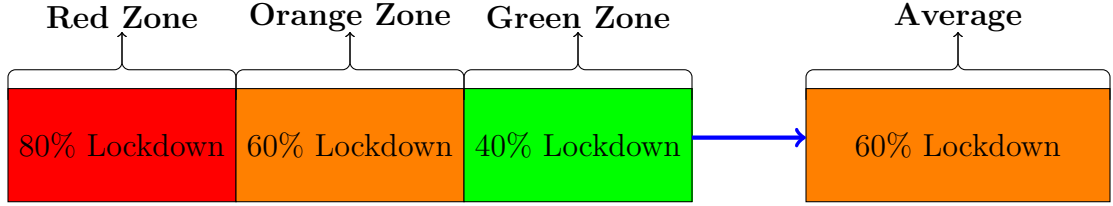
\begin{figure}
\centering
\begin{tikzpicture}
\path
(0,0)       node[onen] (N) {80\% Lockdown}
++(0:3)    node[twoo] (C) {60\% Lockdown}
+(0:3)    node[thrree] (O) {40\%  Lockdown}
++(0:8)    node[fouur] (A) {60\% Lockdown};
\draw[->,rounded corners=1mm] (-1.5,0.6) |- (-0.0,1) -- ++(0,0.5);
\draw[->,rounded corners=1mm] (1.5,0.6) |- (-0.0,1) -- ++(0,0.5);
\draw[->,rounded corners=1mm] (1.5,0.6) |- (3.0,1) -- ++(0,0.5);
\draw[->,rounded corners=1mm] (4.5,0.6) |- (3.0,1) -- ++(0,0.5);
\draw[->,rounded corners=1mm] (4.5,0.6) |- (6,1) -- ++(0,0.5);
\draw[->,rounded corners=1mm] (7.5,0.6) |- (6,1) -- ++(0,0.5);
\draw[->, blue, ultra thick](7.5,0.1) -- (9.0,0.1);
\draw[->,rounded corners=1mm] (9.0,0.6) |- (11.0,1) -- ++(0,0.5);
\draw[->,rounded corners=1mm] (13.0,0.6) |- (11.0,1) -- ++(0,0.5);
\node at (0.0,1.7) {\bf Red Zone};
\node at (3.0,1.7) {\bf Orange Zone};
\node at (6.0,1.7) {\bf Green Zone};
\node at (11.0,1.7) {\bf Average};
\end{tikzpicture}
\caption{Zone-wise lockdown and its averaging(Arithmetic Mean) is a part of LD-I} \label{fig:zoneld}
\end{figure}

\noindent Similarly, one may interpret the policies LD-II and LD-III in Table \ref{table:peak}. LD-I policy may be referred as a strict lockdown policy implemented for a prolonged time period, whereas, policies LD-II and LD-III may be pertained to an on-off type lockdown exit strategy as discussed earlier. We analyse the impact of these lockdown policies in containing the disease transmission in India by incorporating these policies in our SEAIRD-control model.\\

\noindent Fig. \ref{fig:2wave} depict the trajectories of asymptomatic and symptomatic infections ($A+I$), hospitalised cases ($H$), total active cases ($A+I+H$), and death cases under the SEAIRD-control model with LD-I policy. The active cases trajectory is seen to have two peaks -- the first peak appear around day 300 and the second peak around day 700. Further, it may be observed that the second peak size is nearly three times bigger as compared to the first one and a gap of nearly 400 days between the two peaks. The advent of the second peak, although after a gap of 400 days, suggests that the policy of strict lockdown for a prolonged time period does not break the virus transmission. Rather, it accumulates a large pool of completely naive susceptible population prone to get eventually infected in a big way. Thus, a policy like LD-I, in addition to incurring heavy socio-economic cost does not prove to be effective in eliminating the infectious disease. On the flip side this policy presents a fairly long time period (400 days) for carrying out vaccine development program and capacity building of the healthcare system. The effective reproduction number in this scenario starts its final descent after 600 days and eventually falls below 1, see Fig. \ref{fig:ERN}.\\

\noindent Fig. \ref{fig:symwave} displays the evolution of symptomatic infections under the SEAIRD-control model influenced by the three intervention policy scenarios, namely, LD-I, LD-II, and LD-III. In all the three scenarios, two infection peaks are observed. The peak size and position of the LD-I influenced trajectory is seen to be in conformity with the trajectories in Fig. \ref{fig:2wave}. With regard to LD-II and LD-III peaks, one may observe that both peaks in each trajectory are nearly same size. The major advantage of LD-III peaks over LD-II peaks is the significant time gap between the two peaks. Therefore, LD-III may be perceived to be an optimal lockdown policy. The effective reproduction number trajectories for all the three scenarios are plotted in Fig. \ref{fig:ERN}.\\ 
\begin{figure}
    \includegraphics[width=\textwidth]{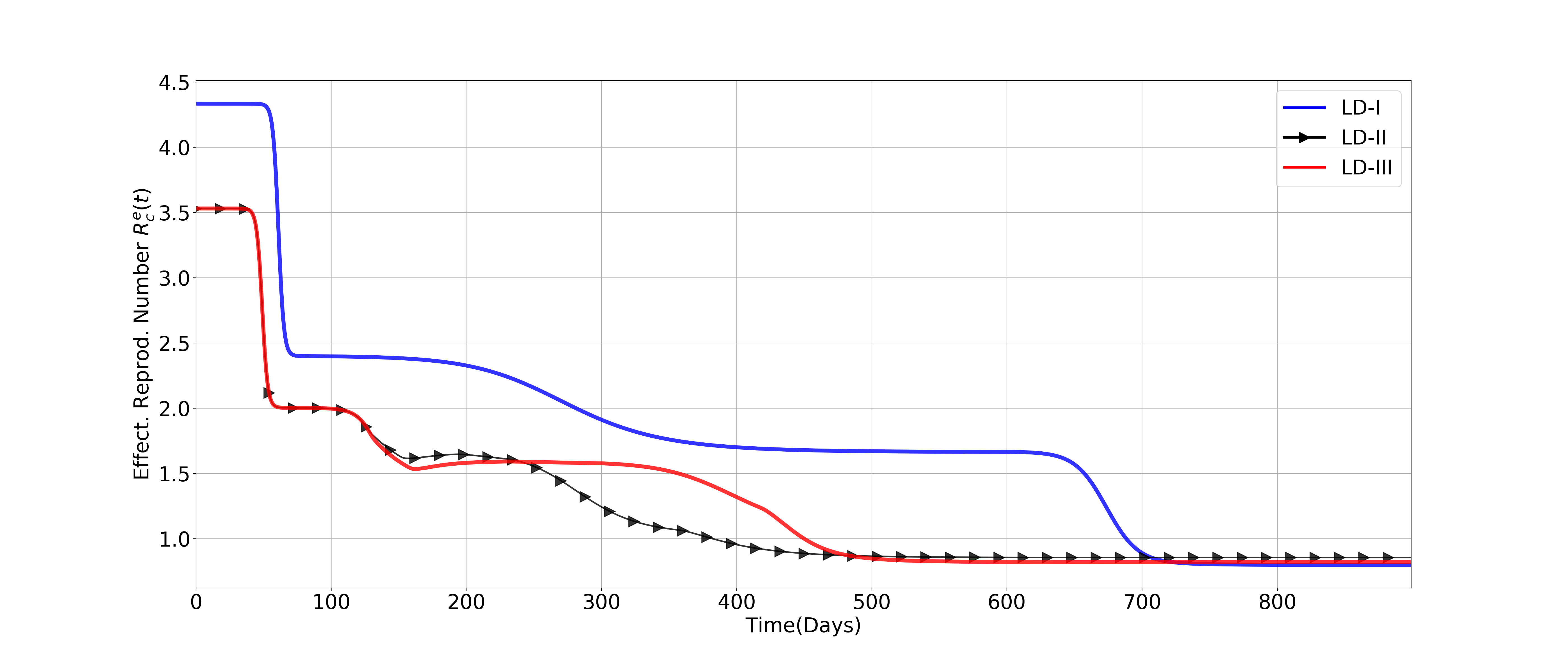}
    \caption{\textbf{SEAIRD-Control model :} Three different effective reproduction numbers corresponding to their lockdown policies and then how they are monotonically decreasing and finally became less than one, when the pandemic has lost its strength for growth.}\label{fig:ERN}
\end{figure}

\noindent The basic reproduction number ($R_0$) of COVID-19 has been initially estimated by the WHO to be in the range 1.4 to 2.5, as declared in the statement regarding the outbreak of SARS-CoV-2, dated January 23, 2020. Later in \cite{R0range:2020, R0covid:2020}, the researchers estimated the mean value of $R_0$ to be higher than 3.28 and median higher than 2.79, by observing the super spreading nature and the doubling rate of the novel Coronavirus. Our effective reproduction numbers are within the estimated range.\\

\begin{table}[!h]
\begin{center}
\caption{Symptomatic count during different lockdown policies}\label{table:symptomatic}
\begin{tabular}{ |p{1.3cm}||p{3.5cm}|p{3.5cm}|p{5.2cm}|  }
\hline
{\bf Policy}& 	{\bf First Peak} 		& 	{\bf Second Peak} 	& 	{\bf Symptomatic (\%) }\\
		& 	Date (No. of Days) 		& 	Date (No. of Days) 		& 	First Peak	 \textemdash \; 	 Second Peak	\\
\hline\hline
LD-I 		& 	28 Dec 2021 (299)	& 	23 Jan 2022 (692) 	& 	0.08\% \textemdash \;  0.29\%\\
\hline\hline
LD-II 	& 	2 Aug 2020 (151)	& 	01 Jan 2021 (303) 	&	 0.15\% \textemdash \; 0.15\% \\
\hline\hline
LD-III 	& 	29 Jul 2020 (147) 	& 	21 May 2021 (451) 	& 	0.175\% \textemdash \; 0.175\% \\
\hline
\end{tabular}
\end{center}
\end{table}

\noindent In Table \ref{table:symptomatic}, we present the position (Date and No. of Days) and size (Symptomatic \%) of both the symptomatic infections peaks achieved in the three lockdown policy scenarios.


\subsection{Discussions and challenges}\label{sec:05:challenge}
\noindent \emph{Discussions:} We have modeled and analysed the impact of general nationwide lockdown policies, quarantine and hospitalization measures, and social distancing practices. It would be interesting to model the effect of local level intervention strategies, like inter-state travel ban, shutdown during weekends, containment zones, multiple levels of location dependent lockdown based on the intensity of infections and different means of social distancing measures. At the outset it appears that lockdowns are being released in phases but on careful observation one may notice that containment zones are still under strict lockdown.
This may look like a large part of the country is free from strict interventions but it would be difficult to deny that on an average the lockdown intensity in India has never been below 60\% which correspond to LD-I policy as described in Section \ref{sec:05}. Especially, if we restrict our domain only to the infected regions of the country, then we may realise that lockdown in India has been strict for long time. 
As per our model prediction the first wave in Fig. \ref{fig:2wave} will attain its peak by the end of December 2020 with over 12 million active cases and a possible second wave with much bigger peak by the end of January 2022. The prediction of a relatively small first peak may alleviate the overwhelming strain on healthcare system of a large country like India. The development of medicine and vaccine can only stop the second wave and its worst impact. \\

\noindent In our model the number of compartments are limited to eight. More realistic models may have hundreds of compartments. With efficient computational algorithms like reduced order model strategies for faster and accurate computations one may arrive at long term reliable prediction. This type of large pandemic model results are more closer to the realistic data and helps in identifying the compartments and parameters having greater influence in the epidemic spread or control. There are several possible compartments that we are interested to consider in our future work, like intra and inter-state travel restrictions, regions where recovery rate is high due to better healthcare facilities, multiple strains of the virus, transmission through droplets, availability of testing kits, incorporating gender structure in the contact matrix, population density in urban and rural India and lockdown releasing strategies based on age and occupation classes.\\

\noindent \emph{Challenges:} There are multiple challenges we have realised during the simulation and analysis of computational results while also targeting to match with the real data. 
\begin{enumerate}[(i)]
    \item The stability result for the SEIARD-control model is difficult to establish due to a large number of parameters involved in it. Theoretical results on stability are open for future work. 
    
    \item Theoretically, there is no proof of how and when the subsequent peaks will appear after the first peak. Based on computational observations the possibility of a second wave is predicted.
    
    \item The model parameters are assumed to be independent of both time and age group. If the model parameters can be estimated corresponding to each age group then predicted numbers may change significantly. 
    
    \item As time progresses containment zones distribution also changes, and due to this, the lockdown intensity varies with time and locations. How efficiently one can measure the effective lockdown of a highly populated country at a particular instant of time and its impact on the predictions?
\end{enumerate}


\section{Conclusions}\label{sec:06}
In this paper, we have proposed a mathematical study of the evolution of COVID-19 in India using the SEAIRD type epidemiological models. An improvised model that takes into consideration various factors such as age-structured social contact pattern, lockdown measures, contact tracing, quarantine, hospitalization, time dependent incidence parameter as another form of social distancing measure has been developed. The proposed model is shown to match with the real COVID-19 data of India, active cases and death cases till May 15, 2020. 
Further, we have studied the impact of three different types of lockdown policies with exit plans over a period of 900 days. Several interesting observations are made in the process of fine tuning the mathematical model towards the twin goals of capturing the realistic phenomena and making credible suggestions to policy makers involved in the epidemic management. Here we summarise the observations.
\begin{enumerate}[(a)]
\item The SEAIRD-control pandemic model parameters are traced back to match the real data for Indian COVID-19 cases.

\item Effective reproduction number is computed for the SEAIRD-control model. Disease free and endemic equilibrium points are proved to be locally stable for the SEAIRD model.

\item In Subsection \ref{sec:03:LD}, we have computationally established that slowly decreasing staggered lockdown is comparatively better lockdown exit strategy in terms of keeping the infection levels low.
 
\item The social distancing function in subsection \ref{sec:04:SD} is another control function incorporated to slowdown the rate of new infection with time and account for behavioural change in the population. 

\item The consequence of implementing strict lockdown for prolonged time period could be the advent of a bigger second wave of infections, albeit after a fairly long time gap and in the event of no progress in potent vaccine development.
 
\item The staggered lockdown policies proposed in Section \ref{sec:05} are designed in such a way that both the waves can be kept under control in order to facilitate hospital treatments for the critical patients.

\item Zone-wise lockdown policy implemented in India has been modeled to match with the real data as accurately as possible. Our model prediction for July 20, 2020 is approximately 436,885 active cases whereas, the real data for the same date is 401,606.
  
\item As per our study, peak of the first wave may arrive no later than the end of December 2020. The policy maker may not have enough time to control the second wave if there is no progress in developing viable medicines or vaccines.

\end{enumerate}

\section{Acknowledgements}
NM~:~ Work of this author was partially funded by grant MTR/2019/001366 of the Science and Engineering Research Board of India.\\
SR~:~ Work of this author was partially funded by grant MSC/2020/000289 of the Science and Engineering Research Board of India.\\
MB \& SN~:~ Acknowledges IIITDM Kancheepuram for their research fellowship and facilities, more importantly during this COVID-19 lockdown period.

\vspace{-1.0cm}
\section*{~~~}

\end{document}